\documentclass[10pt]{article}
\RequirePackage{anyfontsize}
\renewcommand\normalsize{%
  \fontsize{10.3pt}{12.6pt}\selectfont}
  \normalsize
  
\usepackage{url}
\usepackage{mathtools}
\usepackage{amssymb}
\usepackage{amsthm}
\usepackage{enumitem}
\usepackage{dsfont}
\usepackage{appendix}
\usepackage{color} 
\usepackage{comment}
\usepackage{stmaryrd}
\usepackage{authblk}
\usepackage[unicode]{hyperref}
\usepackage[utf8]{inputenc}
\usepackage[T1]{fontenc}
\usepackage{geometry}
\usepackage{float}
\usepackage{bm}
\usepackage{algpseudocode, algorithm}
\usepackage{subcaption}
\usepackage{caption}
\usepackage{graphicx}
\usepackage{booktabs}
 
\definecolor{red}{rgb}{0.7,0.15,0.15}
\definecolor{green}{rgb}{0,0.5,0}
\definecolor{blue}{rgb}{0,0,0.7}
\hypersetup{colorlinks, linkcolor={red},citecolor={green}, urlcolor={blue}}
			
\newtheorem{theorem}{Theorem}[section]

\newtheorem{lemma}[theorem]{Lemma}
\newtheorem{proposition}[theorem]{Proposition}

\newtheorem{definition}[theorem]{Definition}
\newtheorem{remark}[theorem]{Remark}

\def \E{\mathbb{E}}
\def \F{\mathbb{F}}
\def \G{\mathbb{G}}

\def \I{\mathbb{I}}

\def \N{\mathbb{N}}

\def \P{\mathbb{P}}

\def \R{\mathbb{R}}

\def \V{\mathbb{V}}

\def\Ac{{\cal A}}

\def\Cc{{\cal C}}
\def\Dc{{\cal D}}

\def\Fc{{\cal F}}

\def\Ic{{\cal I}}
\def\Jc{{\cal J}}

\def\Lc{{\cal L}}
\def\Mc{{\cal M}}

\def\Oc{{\cal O}}

\def\Uc{{\cal U}}

\def\setK{{\llbracket 1,K\rrbracket}}
\def\setKzero{{\llbracket 0,K\rrbracket}}

\def\1{\mbox{1\hspace{-0.25em}l}}
\newcommand{\ud}{\mathrm d}
\newtheorem{Assumption}{Assumption}[section]

\def\d{\mathrm{d}}
\newtheorem{nota}{Notation}[section]

\DeclareMathSymbol{\shortminus}{\mathbin}{AMSa}{"39}

\title{
\Large
Trading in CEXs and DEXs with Priority Fees and Stochastic Delays
\thanks{PB gratefully acknowledges support from
the ILB Chair Artificial Intelligence and Quantitative Methods for Finance at University Paris Dauphine-PSL. YH acknowledges support from  the Chaire Risque Financiers, Société Générale, at École Polytechnique, Institut Europlace de Finance, and the Oxford-Man Institute of Quantitative Finance (OMI). This work was carried out while YH was visiting LSB at the OMI. We thank \'Alvaro Cartea and Fayçal Drissi for comments on an earlier draft. }
}

\author[1]{\normalsize Philippe Bergault}
\author[2]{\normalsize Yadh Hafsi}

\author[3,4]{\normalsize Leandro S\'anchez-Betancourt} 
\affil[1]{CEREMADE, Université Paris Dauphine-PSL}
\affil[2]{CMAP, École Polytechnique}
\affil[3]{Mathematical Institute, University of Oxford}
\affil[4]{Oxford-Man Institute of Quantitative Finance, University of Oxford}

\begin{document}

\maketitle
\vspace{-0.3cm}
\begin{abstract}
    We develop a mixed control framework that combines absolutely continuous controls with impulse interventions subject to stochastic execution delays. 
    The model extends current impulse control formulations by allowing (i) the controller to choose the mean of the stochastic delay of their impulses, and allowing (ii) for multiple pending orders, so that several impulses can be submitted and executed asynchronously at random times. 
    The framework is motivated by an optimal trading problem between centralized (CEX) and decentralized (DEX) exchanges. In DEXs, traders control the distribution of the execution delay through the priority fee paid, introducing a fundamental trade-off between delays, uncertainty, and costs. 
    We study the optimal trading problem of an agent exploiting trading signals in CEXs and DEXs. 
    From a mathematical perspective, we derive the associated dynamic programming principle of this new class of impulse control problems, and establish the viscosity properties of the corresponding quasi-variational inequalities. From a financial perspective, our model provides insights on how to carry out  execution across CEXs and DEXs, highlighting how traders  manage latency risk optimally through priority fee selection.
    We show that employing the optimal priority fee has a  significant outperformance over non-strategic fee selection.
\end{abstract}\hspace{10pt}
\\
 \noindent\textbf{Keywords:} Impulse control with delay, priority fees, optimal trading, mixed-control, viscosity solutions, decentralised finance, automated market makers.

\section{Introduction}

The interplay between trading speed, execution uncertainty, and market structure has become increasingly prominent with the rise of decentralized exchanges (DEXs). Currently, within the crypto space, DEXs operate alongside traditional centralized exchanges (CEXs). In DEXs, traders face a fundamental trade-off between immediacy and execution risk, as the timing (position in queue) of order execution is subject to a stochastic delay that can be shortened by offering a ``priority fee'' to the miners. This motivates the need for the development of mathematical tools capable of capturing decision-making under delayed (and to some extent controlled) uncertain executions. 
In this paper, we develop a theory for a new type of impulse control problem with stochastic delay in which the controller 
influences the stochastic delay of the execution of their impulses through a control variable. 
We extend the literature in two key directions: (i) enabling control over the mean of the stochastic delay of the impulses, and (ii) allowing for multiple pending orders within the impulse control with stochastic delay framework. From a financial perspective, our framework provides a novel formulation for optimal trading between CEXs and DEXs, where the trader strategically selects priority fees to manage execution speed and risk.\\

A detailed summary of our contributions is as follows.
Mathematically, we introduce a framework that handles multiple pending orders within the class of impulse control problems with stochastic delay. Furthermore, we allow the controller to choose the mean of the random delay of each of their impulses. We combine a regular control, which models continuous trading on the CEX, with an impulse control, which models discrete order submissions to the DEX.  On the impulse side,  \cite{oksendal2005,Oksendal2008,BRUDER20091436,cartea2023optimal}, investigate optimal stopping/impulse control problems with deterministic delay. In \cite{Oksendal2008}, an infinite-horizon setting is considered with deterministic delay and an arbitrary number of pending orders, while \cite{BRUDER20091436} treats a finite-horizon problem with any finite number of pending orders. In these deterministic-delay models, the main technical burden comes from enforcing the associated boundary and consistency conditions induced by the fixed delay.\footnote{By contrast, \cite{cartea2023optimal} consider stochastic execution delay where the delay mechanism is exogenous and not a control variable. The controller chooses order submission times and sizes, while execution latency is governed by a given Poisson process, leading to a different structure in the intervention operator. } 
Closest to our work is \cite{cartea2023optimal}, where the authors introduced an impulse control problem with uncontrolled stochastic delays and one pending order. Here, on the other hand, we control the mean of the stochastic delay and we allow for a finite number of pending orders.  
Regarding the option of the controller to choose the mean of the stochastic delay, the only work that we are aware of is  \cite{becherer2023mean}, who  deal with a Markov decision process and a discrete-time version of a controlled deterministic delay (in their framework this is called observation delay). As far as we are aware, our framework is the first to address optimal execution with random latency in which the trader can influence the mean of the  latency distribution.\\

From a financial point of view and to the best of our knowledge, our work is the first to study optimal trading between centralised exchanges (CEX) and decentralised exchanges (DEX) while employing the key degree of freedom that liquidity takers have when choosing the ``priority fee'' of their orders. We find the optimal decision boundaries that traders should employ when selecting the priority fee attached to their orders. Furthermore, we show how these regions change with time, inventory, and price discrepancies. 
As expected, the performance of the trader increases with the number of priority fees they consider, but we find that this performance plateaus fairly quickly, which implies that entertaining a finite number of priority fees  is close to optimal. 
Within the CEX-DEX trading context, closest to our work are \cite{bichuch2024defi,cartea2023execution,cartea2025decentralised,fukasawa2024model,he2025arbitrage,jaimungal2023optimal}. However, none of these works accounted for execution delays and the crucial role of the priority fees. To the best of our knowledge, this is the first continuous-time optimal control framework to tackle the problem of trading in CEX-DEXs accounting for both priority fees and execution delays.\footnote{Priority fees have been studied in \cite{capponi2026price,capponi2025longer} within the context of price discovery and market efficiency. More precisely, \cite{capponi2026price} introduce a competition model for trading in blockchains with priority fees and its impact on price discovery, and \cite{capponi2025longer} investigates whether priority fees carry information (or market impact). On the other hand,  \cite{hasbrouck2022need,aqsha2025equilibrium,baggiani2025optimal, campbell2025optimal} study the role of fees  in attracting order flow to the venue.  }\\

The remainder of the paper proceeds as follows. Section $\ref{sec:problem_formulation}$ formulates the mixed control problem with (controlled) stochastic execution delay. Section $\ref{sec:viscosity}$ establishes the dynamic programming principle and characterizes the value function via the associated Hamilton–Jacobi–Bellman quasi-variational inequality, showing that it is the unique viscosity solution of the HJBQVI. Section $\ref{sec:CEX-DEX}$ applies the framework to an optimal trading problem in which the agent trades continuously on a centralized exchange and discretely on a decentralized exchange. Lastly, we provide numerical illustrations and study how the optimal strategy varies across model specifications and parameter choices.

\paragraph*{Notation.}
For $x \in \mathbb{R}^k$, where $k$ is determined by the context, $\|x\|$ denotes its Euclidean norm, and $B_r(x)$ represents the open ball centred at $x$ with radius $r > 0$. The scalar product is denoted by $\langle \cdot, \cdot \rangle$, and for a vector $x \in \mathbb{R}^k$, its transpose is denoted by $x^\top$. For a set $A \subset \mathbb{R}^k$, $\partial A$ denotes its boundary. For a function $\varphi : \mathbb{R}_+ \times \mathbb{R}^d \times \mathbb{R}^k \to \mathbb{R}$, the gradient and Hessian matrix are denoted by $D\varphi$ and $D^2\varphi$, respectively, whenever they are well-defined. Lastly, $\mathbb{N}^\star = \{1,2,3,\dots\}$.

\section{Problem Formulation}
\label{sec:problem_formulation}
We consider a trading horizon $T>0$ and a complete filtered probability space $(\Omega , \mathcal{F} = \{\mathcal{F}_t\}_{t \geq 0},\mathbb{P})$, where $\{\mathcal{F}_t\}_{t \geq 0}$ is a right-continuous filtration. We assume that the filtration $\F$ supports a $q$-dimensional Brownian motion $W$ together with a collection of point processes $(N^i)_{i\in \mathcal{I}}$ with $\mathcal{I} =\llbracket 1,N \rrbracket$, where $N\in\mathbb N^\star$.\\

For each $i\in\llbracket 1,N\rrbracket$, the sequence $(T^i_n)_{n\geq 1}$ 
defines a Poisson process with intensity $\ell_i>0$. We take
$$N^i_t = \sum_{n=0}^{+\infty} \mathds{1}_{\{T^i_n \le t\}},\quad \forall i \in \Ic,$$
with convention $N^i_{0-} = 0$.  
Let $\G^{0}$ be the filtration associated to $W$ and the collection of $N$ point processes $(N^i)_{i\in \Ic}$ such that 
$$\mathcal G^0_t
:= \sigma\Big(
(W_s)_{0\le s\le t},\;
(N^i_s)_{0\le s\le t}:\; i=\{1,\dots,N\}
\Big).$$
The natural filtration associated with $W$ and 
$(N^i)_{1\le i\le N}$ is then given by the usual augmentation
$$\mathcal G_t := \bigcap_{u>t} \big( \mathcal G_u^{\,0} \vee \mathcal N_0 \big),
\quad t\ge0,$$
where $\mathcal N_0$ denotes the $\mathbb P$-null sets of $\mathcal F$. \\

Here $N$ represents the number of available priority fees. In practice, agents transacting on a blockchain may choose any positive priority fee to attach to a given action (for instance, a swap on an AMM or a simple transfer). For tractability, we restrict this choice to a finite set of $N$ possible priority fee levels. This discretization is a mild modelling simplification: it does not alter the qualitative behaviour of the system nor the nature of the results, but it allows us to work with a well-defined and finite family of point processes.\\

We work in a setup with up to $K$ pending orders and with each order the controller may choose one of the $N$ expected delays with an associated fee. 
Each fee level corresponds to an expected execution delay, although the actual execution time is random.  
For $i\in\{1,\dots,N\}$, the cost of priority fee $i$ is denoted by $\mathfrak p_i>0$, and the associated expected execution delay is $\ell_i>0$.  
For convenience we assume that if $i<j$ then $\mathfrak p_i<\mathfrak p_j$ and $\ell_i>\ell_j$.  
We define the fee vector $\mathfrak p=\{\mathfrak p_1,\dots,\mathfrak p_N\}$ and the delay vector $\mathfrak{l}=\{\ell_1,\dots,\ell_N\}$. \\

To describe the agent’s discrete trading decisions, we introduce an impulse control $$(\tau_n, I_n, \xi_n)_{n \ge 1},$$ 
which specifies the sequence of intervention times, priority fee indexes, and order sizes. Formally, $(\tau_n, I_n, \xi_n)_{n \ge 1}$ consists of a non-decreasing sequence of $\G$-stopping times (intervention times) $\tau_n\le T$, priority indexes  $I_n\in \mathcal{I}$, and impulse actions $\xi_n\in \Uc:= [ -\hat{V}, \hat{V}]$ for $\hat{V}>0$. These impulse actions represent the volume to be executed (positive for buys orders and negative for sells). For a given $(\tau_n, I_n)$, let $m$ (which depends on $\tau_n$) be such that $T^{I_n}_{m-1}\leq \tau_n < T^{I_n}_{m}$, then,  define $\tilde{\tau}_n = T^{I_n}_{m}$. In words, $\tilde{\tau}_n$ is the next time after $\tau_n$ where we observe a jump from the  point process $N^{I_n}$. In what follows we use the notation $\Tilde{\cdot}$ to denote execution times.\\

In addition to impulse decisions, the agent controls an absolutely continuous trading rate. We denote by $\nu=(\nu_t)_{t\in[0,T]}$ a $\G$-progressively measurable c\`adl\`ag process such that
$$\E\Big[\int_0^T \nu_s^2\,ds\Big]<+\infty.$$
We interpret $\nu_t$ as the signed continuous trading speed at time $t$ on the centralized exchange (CEX). Formally, the agent's control is the pair
$$\alpha := \Big((\nu_t)_{t\in[0,T]},(\tau_n,I_n,\xi_n)_{n\ge1}\Big).$$

We introduce the $\G$-adapted process $\iota(\cdot,\alpha)$, such that $\iota(t,\alpha)$ returns the ordered indexes of the 
pending orders at time $t$ under strategy $\alpha$
$$\iota(t,\alpha) := \{n\ge1 : \tau_n \le t < \tilde\tau_n \}.$$
Let $k(\cdot,\alpha)$ be a $\G$-adapted process defined by $k(t,\alpha) := \operatorname{card}(\iota(t,\alpha))$,
so that $k(t,\alpha)\in\llbracket 0,K\rrbracket$ represents the number of pending actions at time $t\in[0,T]$. 
We define the pending orders as 
$$\mathfrak{P}(t,\alpha) := (I_i,\xi_i)_{i\in \iota(t,\alpha) }.$$
We introduce the $K$-dimensional vectors representing, respectively, the number of pending orders at each priority level and the associated pending volumes, defined by
\begin{equation*}
\resizebox{\textwidth}{!}{$
\begin{aligned}
\mathfrak{i}(t,\alpha)  = \bigg(\sum_{i\in \iota(t,\alpha) }\mathds{1}_{\{I_i = 1\}}, \dots,   \sum_{i\in \iota(t,\alpha) }\mathds{1}_{\{I_i = K\}}\bigg)~~\text{and}~~\mathfrak{v}(t,\alpha) = \bigg(\sum_{i\in \iota(t,\alpha) }\xi_i\,\mathds{1}_{\{I_i = 1\}}, \dots,   \sum_{i\in \iota(t,\alpha) }\xi_i\,\mathds{1}_{\{I_i = K\}}\bigg).
\end{aligned}
$}\end{equation*}
In the above, we use the convention that $\sum_{i\in\emptyset}$ is zero. Lastly, we let
$$
 p(t,\alpha)= \big(\mathfrak{i}(t,\alpha),\mathfrak{v}(t,\alpha)\big)\,.
$$
Therefore, the set of admissible strategies for $K$ pending orders is 
\begin{equation*}
\resizebox{\textwidth}{!}{$
\begin{aligned}
    \mathcal{A}_K = \bigg\{ 
    \alpha = \big((\nu_t)_{t\in[0,T]},( \hspace{-5em}\underbrace{\tau_n, I_n, \xi_n}_{\text{intervention time, priority  index, volume}} \hspace{-5em} )_{n\geq 1}\big) :\,  &\text{for }n\geq 1,\,   I_n\in\mathcal{I},\,\,\xi_n \in \Uc,\,\nu~\text{is càdlàg $\G$-progressively measurable,}\\
    \E\bigg[\int_0^T\nu^2_s \d s\bigg]<+\infty&,\;k(t,\alpha)\leq K \text{ for all }t\in[0,T], \text{ and }\tau_n \text{ are ordered $\G$-stopping times}
\bigg\}\,.
\end{aligned}
$}\end{equation*}

With a slight abuse of notation, the set of admissible strategies at time $t\in [0,T)$ is 
\begin{align}
    \mathcal{A}_{K}(t) = \Big\{
   \alpha=\big((\nu_s)_{s\in[t,T]},(\tau_n, I_n, \xi_n)_{n\geq 1}\big) \in \mathcal{A}_K \,:\, \tau_1 \geq t 
    \Big\}\,.
\end{align}

\begin{lemma}
\label{memoryless_prop}
    Let $\alpha=\big((\nu_t)_{t\in[0,T]},( \tau_n, I_n, \xi_n)_{n\geq 1}\big)\in\mathcal{A}_K$. The following two properties hold.
    \begin{enumerate}
        \item $(\tilde{\tau}_n)_{n\in\N}$ are $\G$-stopping times.
        \item $(\tilde{\tau}_n - \tau_n)_{n\in\N}$ are a collection random variables such that $\tilde{\tau}_n - \tau_n$ is exponentially distributed with parameter $\ell_{I_n}$.
    \end{enumerate}
\end{lemma}
\begin{proof}
    The proof follows from Lemma 3.1 in \cite{cartea2023optimal}. For part (ii) the result follows from the memoryless property of exponential random variables.
\end{proof}

The  objective of the agent is to trade optimally over the time horizon $[0,T]$. We introduce the set 
\begin{equation}
    \mathbb{I}_m = \bigg\{(i_1,\dots,i_N)\in \setKzero^N \,:\, \sum_{j=1}^N i_j = m   \bigg\}\,,
\end{equation}
which gathers  all possible configurations of $m\in\setKzero$ pending orders per priority index. 
Similarly, 
\begin{equation}
 \mathbb{V}_1 = \Big\{(v_1,\dots,v_N)\in \mathcal{U}^N \,:\, \exists j\in \llbracket 1,N\rrbracket \text{ s.t. } v_j\neq0 \text{ and }  v_i = 0 \text{ for }i\neq j,\quad i\in\llbracket 1,N\rrbracket \Big\}\,,
\end{equation}
and for $m\in \{2,\dots,K\}$, we define $\mathbb{V}_m$ recursively as
\begin{equation}
%\mathbb{I}_m =  \underbrace{\mathbb{I}_1 + \dots + \mathbb{I}_1}_{m \text{ times}}\,,\qquad 
\mathbb{V}_m =  \mathbb{V}_{m-1} +  \mathbb{V}_1 = \big\{v+w\,:\, v\in\mathbb{V}_{m-1},\,\, w\in\mathbb{V}_1 \big\}\,.
\end{equation} 
As a convention, we set $\mathbb{I}_0=\varnothing$, $\mathbb{V}_0=\varnothing$. We define the sets $\I$ and $\V$ as
$$\I := \underset{0\le m\le K}{\bigcup} \mathbb I_m~~\text{and}~~ \mathbb  V := \underset{0\le m\le K}{\bigcup} \mathbb V_m.$$ 
Let $P$ denote the function
\begin{align*}
    P(t,u,\mathfrak i) 
    &= \sum_{i=1}^N \mathfrak i_i \left( 1 - \mathds 1_{\{N^i_u > N^i_t\}} \right),\quad \forall\; 0\le t \le u \le T, \; \mathfrak i\in \I.
\end{align*}
 Intuitively, if the vector $\mathfrak i$ encodes the set of pending orders at time $t$, then $P(t,u,\mathfrak i)$ represents the number of those initial orders that remain pending at time $u \ge t$.
We also define the pending volume to be 
\begin{align*}
    V(t,u,\mathfrak v) 
    &= \sum_{i=1}^N |\mathfrak{v}_i| \left( 1 - \mathds 1_{\{N^i_u > N^i_t\}} \right), \quad \forall\; 0\le t \le u \le T, \; \mathfrak v\in \mathbb  V.
\end{align*}
Based on these quantities, we introduce the admissible control set
\begin{equation*}
\resizebox{\textwidth}{!}{$
\begin{aligned}\mathcal A_{K, \mathfrak i, \mathfrak v}(t)= \Big\{ \alpha=\big((\nu_s)_{s\in[t,T]},(\tau_n, I_n, \xi_n)_{n\geq 1}\big) \in \mathcal{A}_{K}(t) : k(s, \alpha) \le K-P(t^-,s, \mathfrak i )~\text{and}~  V(t,s,\mathfrak v)  \leq \bar{V},~ \forall  s\in [t,T]\Big\},\end{aligned}
$}\end{equation*}
with the convention $\mathcal A_{K, \mathfrak i, 0}(0):= \mathcal A_K $, for all $\mathfrak i\in \I$. The quantity $\bar V$ is a trading constraint; such constraint is not necessarily the inventory because we are not restricted to an optimal liquidation problem. The condition $V(t,s,\mathfrak v)  \leq \bar{V}$ establishes that the pending volume waiting to be executed in the AMM cannot exceed $\bar{V}$.\\

At execution $\tilde\tau_n$, the controlled càdlàg state $X=(X_t)_{t\in[0,T]}$ jumps according to a measurable impulse map $\Gamma$. For $t\in[0,T]$ and a given policy $\alpha\in \Ac_{K}(t)$, the dynamics of the  $d$-dimensional state process $X$ are given by
\begin{equation}\label{eq:controlledJD}
\begin{aligned}
X^{t,x,\alpha}_u &= x + \int_t^u b(s, X^{t,x,\alpha}_s,\nu_s)\,\d s + \int_t^u \sigma(s, X^{t,x,\alpha}_s,\nu_s)\,\d W_s+ \sum_{t\le \tilde\tau_n\le u}\big(\Gamma(\tilde\tau_n,X^{t,x,\alpha}_{\tilde\tau_n^-},\xi_n)-X^{t,x,\alpha}_{\tilde\tau_n^-}\big),
\end{aligned}
 \end{equation}
 for $t\in[0,T]$, with measurable coefficients $b:[0,T]\times \R^d\times\R\to\R^d$, $\sigma:[0,T]\times\R^d\times\R\to\R^{d\times q}$, 
$\Gamma:[0,T]\times\R^d\times \Uc\to\R^d$.\\
 
Let $t \in [0,T]$ and $(x,\mathfrak{i},\mathfrak{v}) \in \Dc$, where $\Dc$ defines the following domain
$$\Dc := \Big\{(x,\mathfrak{i},\mathfrak{v}) : x \in \R^d ,\; 
\mathfrak{i} \in \mathbb{I},\; \mathfrak{v} \in \mathbb{V},\; \sum_i |\mathfrak{v}_i|<\bar{V}\Big\}.$$
Here, the elements $(\mathfrak{i},\mathfrak{v}) = (0_{\mathbb{I}},0_{\mathbb{V}})$ denote the absence of pending orders. Let $f:[0,T]\times\R\times\R^d\to\R$ be a running reward, $g:\R^d\to\R$ a terminal payoff, and 
$c:[0,T]\times\R^d\times \Uc\times \I\to\R_+$ an intervention cost. For any admissible control $\alpha \in \Ac_{K, \mathfrak i, \mathfrak v}(t)$, we define the performance criterion
\begin{equation}
  J(t,x,\alpha) = 
  \E\bigg[\int_t^T f(s,X^{t,x,\alpha}_s,\nu_s)\,\d s + g(X^{t,x,\alpha}_T) 
 - \sum_{n\geq 1:\,\tilde{\tau}_n\in(t,T]} c(\tilde{\tau}_n,X^{t,x,\alpha}_{\tilde{\tau}^-_n},\xi_n,I_n)\bigg].
\end{equation}
The associated value function is then given by
\begin{equation}\label{eq:value_function}
    v(t,x,\mathfrak{i},\mathfrak{v}) := \sup_{\alpha \in \Ac_{K, \mathfrak i, \mathfrak v}(t)} J(t,x,\alpha), 
    \quad \forall(t,x,\mathfrak{i},\mathfrak{v}) \in [0,T] \times \Dc.
\end{equation}
In particular, when there are no pending orders, the state reduces to 
$(x,0_{\mathbb{I}},0_{\mathbb{V}}) \in \Dc$, and the value function defined in \eqref{eq:value_function} coincides with the one obtained under  $\mathcal{A}_{K,0,0}(t)$. 

\begin{Assumption}\label{assumptions} We assume that the following holds.
\begin{itemize}
\item[(A1)] 
The maps $b$, $\Gamma$, $f$, $g$ and $c$ are Borel measurable. Moreover, $c$ is continuous, nonnegative, 
and verifies $c(t,x,0,i)=0$ for all $(t,x,i)\in [0,T]\times\R^d\times\llbracket 1,N\rrbracket$. 

\item[(A2)] 
There exists $L>0$ such that, for all 
$t\in[0,T]$, $(a,a')\in\R^2$, $(x,x')\in(\R^d)^2$ and $\xi\in \Uc$,
$$\|b(t,x,a)-b(t,x',a')\|+\|\sigma(t,x,a)-\sigma(t,x',a')\|+\|\Gamma(t,x,\xi)-\Gamma(t,x',\xi)\|
\le L(\|x-x'\| + L|a-a'|).$$

\item[(A3)] 
There exists $C_0>0$ such that, for all $(t,x,a,\zeta)\in[0,T]\times\R^d\times\R\times \Uc$ and all $i\in \llbracket 1,N\rrbracket$,
$$\|b(t,x,a)\|+\|\sigma(t,x,a)\|+\|\Gamma(t,x,\xi)\| + |c(t,x,\xi,i)|
\le C_0\big(1+\|x\|+|a|+\|\xi\|\big).$$

\item[(A4)] There exists $L_f,L_g>0$ such that, for all $t\in[0,T]$, $x,x'\in\R^d$ and $a\in \R$,
$$|f(t,x,a)-f(t,x',a)| \le L_f(1+|a|)\,\|x-x'\|~~\textrm{and}~~
|g(x)-g(x')|\le L_g\,\|x-x'\|.$$

\item[(A5)] For all $(t,\mathfrak i, \mathfrak v)\in[0,T]\times\I\times\V$, we have that 
$$
\sup_{\alpha\in\Ac_{K,\mathfrak i,\mathfrak v}(t)}
\E\left[\int_t^T |f(s,0,\nu_s)|\,\d s\right] \le +\infty.
$$
\item[(A6)] There exists $\lambda>0$ such that, for all $(t,x,a)\in[0,T]\times\R^d\times\R$ and all $\zeta\in\R^d$,
$$\zeta^\top \sigma\sigma^\top(t,x,a)\,\zeta \;\ge\; \lambda\,\|\zeta\|^2.$$

\end{itemize}
\end{Assumption}
In the remainder of this work,  the assumptions in \ref{assumptions} hold. We now state the existence and uniqueness result for the controlled process $X$.
 \begin{proposition}
\label{existence_uniqueness_SDE}
Let $t\in [0,T]$. For any $\Fc_t$-measurable random variable $\xi$ valued in $\R_+$ such that $\E(|\xi|^p)<+\infty$, for some $p>1$, the SDE \eqref{eq:controlledJD} admits a unique strong solution $X^{t,\xi,\alpha}$ under the assumptions (A1)--(A5). Moreover, there exists $C_T>0$ such that
\begin{equation*}
   \E\bigg[\underset{0\leq u\leq T}{\sup}\|X^{t,\xi,\alpha}_u\|^p\bigg]\leq C_T\big(1+\E(|\xi|^p)\big).
\end{equation*}
\end{proposition}
\begin{proof}
    The proof is standard and it follows closely the arguments in \cite[Chapter~V]{protter2005stochastic}.
\end{proof}
We conclude this section by stating the growth bounds satisfied by the value function.
\begin{lemma}[Quadratic growth]
\label{lem:growth_value_function_LQ} There exists positive constants $C_1,C_2>0$ such that, for all
$(t,x,\mathfrak i,\mathfrak v)\in[0,T]\times\Dc$,
\begin{equation}
\label{eq:v_growth_bound_with_f0}
C_1\bigl(1+\|x\|\bigr)\;\le\;|v(t,x,\mathfrak i,\mathfrak v)|
\;\le\;
C_2\bigl(1+\|x\|^2\bigr).
\end{equation}
\end{lemma}
\begin{proof}
    The result follows from the growth conditions on $f$, $g$ and $c$, the moment bounds for the solution of the SDE \eqref{eq:controlledJD} from Proposition \ref{existence_uniqueness_SDE}, the finite jump activity of $(N^i)_{i\in\llbracket 1,K\rrbracket}$ and the boundedness of $\xi$.
\end{proof}

\section{Viscosity Characterization of the Value Function}
\label{sec:viscosity}
In what follows we drop the argument of a mathematical object if it is the empty set.
\subsection{Dynamic programming principle}
\begin{theorem}[Dynamic Programming Principle]\label{thm:DPP}
 The value function $v$ defined in \eqref{eq:value_function} satisfy the dynamic programming principle (DPP). 
That is, for $(t,x)\in [0,T]\times \Dc$, and $\tau \in\mathcal T_{t,T}$, which is the set of $\G$-stopping times valued in $[t,T]$,  we have that 
\begin{equation}
v(t,x,\mathfrak i, \mathfrak v) = \sup_{\alpha\in\Ac_{K, \mathfrak i, \mathfrak v}(t)}
\mathbb E\Bigg[
\int_t^{\tau} f\big(s,X_s^{t,x,\alpha},\nu_s\big)\,\mathrm ds
- \sum_{\tilde \tau_n \in (t,\tau]} 
c(\tilde{\tau}_n,X^{t,x,\alpha}_{\tilde{\tau}^-_n},\xi_n,I_n)
+ v\left(\tau,
X_\tau^{t,x,\alpha},
p(\tau,\alpha)\right)
\Bigg].
\end{equation}
In other words, the following two statements hold. 
\begin{enumerate}
\item \textbf{(DPP1)} For all $\alpha\in \Ac_{K, \mathfrak i, \mathfrak v}(t)$ and for all stopping times $\tau$ valued in $[t,T]$,
$$v(t,x,\mathfrak i, \mathfrak v) \geq 
\mathbb E\Bigg[
\int_t^{\tau} f\big(s,X_s^{t,x,\alpha},\nu_s\big)\,\mathrm \ud s
- \sum_{\tilde \tau_n \in (t,\tau]} 
c(\tilde{\tau}_n,X^{t,x,\alpha}_{\tilde{\tau}^-_n},\xi_n,I_n)
+ v\left(\tau,
X_\tau^{t,x,\alpha},
p(\tau,\alpha)\right)
\Bigg].$$

\item \textbf{(DPP2)} For every $\varepsilon>0$, there exists $\alpha^\varepsilon\in\Ac_{K, \mathfrak i, \mathfrak v}(t)$ such that for all stopping times $\tau$ valued in $[t,T]$,
$$v(t,x,\mathfrak i, \mathfrak v) - \varepsilon \leq 
\mathbb E\Bigg[
\int_t^{\tau} f\big(s,X_s^{t,x,\alpha^\varepsilon},\nu_s\big)\,\mathrm \ud s
- \sum_{\tilde \tau_n \in (t,\tau]} 
c(\tilde{\tau}_n,X^{t,x,\alpha^\varepsilon}_{\tilde{\tau}^-_n},\xi_n,I_n)
+ v\left(\tau,
X_\tau^{t,x,\alpha^\varepsilon},
p(\tau,\alpha^\varepsilon)\right)
\Bigg].$$
\end{enumerate}
\end{theorem}
\begin{proof}
(1) Let $(t,x,\mathfrak{i},\mathfrak{v})\in [0,T]\times\Dc$, $\varepsilon > 0$, $\tau\in\mathcal T_{t,T}$, and $\alpha\in\mathcal A_{K,\mathfrak{i},\mathfrak{v}}(t)$. For any $\omega \in \Omega$, there exists an $\varepsilon$-optimal control $\alpha^{\varepsilon,\omega} \in \Ac_{K, p(\tau(\omega),\alpha(\omega))}(\tau(\omega))$ for 
$v$ at 
$$\big(\tau(\omega),X^{t,x,\alpha(\omega)}_{\tau(\omega)},
p(\tau(\omega),\alpha(\omega))\big)$$ 
such that 
$$v\big(\tau(\omega),X^{t,x,\alpha(\omega)}_{\tau(\omega)},
p(\tau(\omega),\alpha(\omega))\big) - \varepsilon 
\leq J\big(\tau(\omega),X^{t,x,\alpha(\omega)}_{\tau(\omega)},\bar{\alpha}^\varepsilon(\omega)\big).$$
By measurable selection arguments (see \cite[Chapter~VII]{BertsekasShreve1978}), 
 there exists 
$\bar{\alpha}^\varepsilon \in \Ac_{K, p(\tau,\alpha)}(\tau)$ 
such that $$\bar{\alpha}^\varepsilon(\omega)=\alpha^{\varepsilon,\omega}(\omega)
\quad \text{for almost all } \omega \in \Omega.$$
In other words, \begin{equation}\label{eq:A2}
v\big(\tau,X^{t,x,\alpha}_\tau,
p(\tau,\alpha)\big) - \varepsilon 
\leq J\big(\tau,X^{t,x,\alpha}_\tau,
\bar{\alpha}^\varepsilon\big).
\end{equation}
We define the strategy $\tilde\alpha^\varepsilon$ as 
\begin{equation*}
\begin{aligned}
        \tilde\alpha^\varepsilon_u = \left\{\begin{array}{ll}
\alpha_u, &\quad \text { if }u<\tau, \\ 
\bar{\alpha}^\varepsilon_u,&\quad\text { otherwise}. 
\end{array}\right.
\end{aligned}
    \end{equation*}
    Note that $\tilde\alpha^\varepsilon\in\Ac_{K, \mathfrak{i},\mathfrak{v}}(t)$. Indeed, progressive measurability follows from the fact that if 
$\alpha$ and $\tilde\alpha$ are progressively measurable and $\tau$ is a stopping time, 
then 
$$u\mapsto \1_{\{u<\tau\}}\alpha_u+\1_{\{u\ge\tau\}}\bar\alpha^{\varepsilon}_u$$
is progressively measurable. Feasibility and integrability hold on $[t,\tau)$ 
by admissibility of $\alpha$, and on $[\tau,T]$ by admissibility of $\bar\alpha^{\varepsilon}$. Additionally, the number of impulses is respected because 
$\bar\alpha^{\varepsilon}\in\mathcal A_{K,p(\tau,\alpha)}(\tau)$. The non-accumulation property is preserved by 
construction (see \cite[Chapter~IX]{Oksendal2008}).
    Now using the law of iterated conditional expectations and \eqref{eq:A2}, we get that
    \begin{equation*}
\resizebox{\textwidth}{!}{$
\begin{aligned}  
    J(t,x, \tilde\alpha^\varepsilon) &\geq\E\bigg[\int_t^\tau f(s,X^{t,x,\tilde\alpha^\varepsilon}_s,\nu_s)\,\d s 
 - \sum_{n\geq 1:\,\tilde{\tau}_n\in(t,\tau]} c(\tilde{\tau}_n,X^{t,x,\tilde\alpha^\varepsilon}_{\tilde{\tau}^-_n},\xi_n,I_n)+v\left(\tau,
X_\tau^{t,x,\tilde\alpha^\varepsilon},
p(\tau,\tilde\alpha^\varepsilon)\right)\bigg]-\varepsilon\\&\geq\E\bigg[\int_t^\tau f(s,X^{t,x,\alpha}_s,\nu_s)\,\d s 
 - \sum_{n\geq 1:\,\tilde{\tau}_n\in(t,\tau]} c(\tilde{\tau}_n,X^{t,x,\alpha}_{\tilde{\tau}^-_n},\xi_n,I_n)+v\left(\tau,
X_\tau^{t,x,\alpha},
p(\tau,\alpha)\right)\bigg]-\varepsilon.\end{aligned}
$}\end{equation*}
The last inequality holds as $\bar{\alpha}$ coincides with $\alpha$ up to time $\tau$. This completes the proof since $\alpha$, $\tau$, and $\varepsilon$ are arbitrary.

(2) Let $(t,x,\mathfrak{i},\mathfrak{v})\in [0,T]\times \Dc$ and $\tau\in\mathcal T_{t,T}$. For $\alpha\in\mathcal A_{K,\mathfrak{i},\mathfrak{v}}(t)$, the law of iterated conditional expectations gives \begin{equation*}   \begin{split}   
    J(t,x, \alpha)
&=\E\bigg[\int_t^\tau f(s,X^{t,x,\alpha}_s,\nu_s)\,\d s 
 - \sum_{n\geq 1:\,\tilde{\tau}_n\in(t,\tau]} c(\tilde{\tau}_n,X^{t,x,\alpha}_{\tilde{\tau}^-_n},\xi_n,I_n) \\&\quad\qquad+ \E\Big[\int_\tau^T f(s,X^{t,x,\alpha}_s,\nu_s)\,\d s 
 - \sum_{n\geq 1:\,\tilde{\tau}_n\in(\tau,T]} c(\tilde{\tau}_n,X^{t,x,\alpha}_{\tilde{\tau}^-_n},\xi_n,I_n) + g(X^{t,x,\alpha}_T)\mid\Fc_\tau\Big]\bigg].\end{split}
\end{equation*}
Using the memoryless property of the exponential delays (Lemma \ref{memoryless_prop}), we obtain the joint Markov property of $\big(X^{t,x,\alpha}_u, p(u,\alpha)\big)_{u\geq t}$. Combined with the pathwise uniqueness of $X^{t,x,\alpha}$ (Proposition \ref{existence_uniqueness_SDE}), this yields
\begin{equation*}   \begin{split}   
    J(t,x, \alpha)
&=\E\bigg[\int_t^\tau f(s,X^{t,x,\alpha}_s,\nu_s)\,\d s 
 - \sum_{n\geq 1:\,\tilde{\tau}_n\in(t,\tau]} c(\tilde{\tau}_n,X^{t,x,\alpha}_{\tilde{\tau}^-_n},\xi_n,I_n)+J\left(\tau,
X_\tau^{t,x,\alpha},\alpha\right)\bigg]\\&\leq\E\bigg[\int_t^\tau f(s,X^{t,x,\alpha}_s,\nu_s)\,\d s 
 - \sum_{n\geq 1:\,\tilde{\tau}_n\in(t,\theta]} c(\tilde{\tau}_n,X^{t,x,\alpha}_{\tilde{\tau}^-_n},\xi_n,I_n)+v\left(\tau,
X_\tau^{t,x,\alpha},
p(\tau,\alpha)\right)\bigg].\end{split}
\end{equation*}
Since the control $\alpha$ is arbitrary, it follows that
\begin{equation*} 
    \resizebox{\textwidth}{!}{$
    \begin{aligned}v(t,x,\mathfrak i, \mathfrak v) \leq 
\sup_{\alpha\in \Ac_{K, \mathfrak i, \mathfrak v}(t)} \mathbb E\Bigg[
\int_t^{\tau} f\big(s,X_s^{t,x,\alpha},\nu_s\big)\,\mathrm \ud s
- \sum_{\tilde \tau_n \in (t,\tau]} 
c(\tilde{\tau}_n,X^{t,x,\alpha}_{\tilde{\tau}^-_n},\xi_n,I_n)
+ v\left(\tau,
X_\tau^{t,x,\alpha},
p(\tau,\alpha)\right)
\Bigg].\end{aligned}$}
    \end{equation*}

\end{proof}
\subsection{PDE Characterization and Comparison Principle}
Building on the problem formulation in Section \ref{sec:problem_formulation}, we consider the following Hamilton-Jacobi-Bellman quasi-variational inequalities for $v$ on the domain $[0,T]\times\Dc$. If $\sum_{i=1}^N\langle \mathfrak i,e_i\rangle<K$, then the  QVI that characterizes the control problem  in \eqref{eq:value_function} is 
\begin{align}\label{HJB1}
\min \bigg\{&-\frac{\partial v}{\partial t}(t,x,\mathfrak i, \mathfrak v)-\sup_{a\in\R}H^a\Big(t,x,\mathfrak i, \mathfrak v,v,\tfrac{\partial v}{\partial x},\tfrac{\partial^2 v}{\partial x^2}\Big)\,,\, \big(v - \mathcal{M}v\big)(t,x,\mathfrak i, \mathfrak v)\bigg\}=0,
\end{align}
where the non-local intervention operator $\mathcal{M}:\mathcal{C}([0,T]\times\Dc) \rightarrow \mathcal{C}([0,T]\times\Dc)$ is defined as
\begin{equation}\label{intervention_operator}
\mathcal{M}\varphi(t,x,\mathfrak i, \mathfrak v)
:= \sup_{(\xi,i)\in\Uc\times\setK}
\varphi(t,x, \mathfrak i + e_i , \mathfrak v + \xi e_i),
\end{equation}
the Hamiltonian $H^a : [0,T]\times\Dc \times \mathcal{C}([0,T]\times\Dc)^3 \rightarrow \mathbb{R}$ is given by
\begin{equation}\label{Hamiltonian}
\begin{split}
&H^a\Big(t,x,\mathfrak i,\mathfrak v,\varphi,
\tfrac{\partial \varphi}{\partial x},
\tfrac{\partial^2 \varphi}{\partial x^2}\Big)
= 
\underbrace{\mathcal L^a \varphi(t,x,\mathfrak i,\mathfrak v)}_{\text{infinitesimal generator}}
\;+\;
\underbrace{\mathcal J \varphi(t,x,\mathfrak i,\mathfrak v)}_{\text{execution of pending orders}}
\;+\;
\underbrace{f(t,x,a)}_{\text{running rewards}},
\end{split}
\end{equation}
the partial differential operator $\mathcal{L}^a:  [0,T]\times\R^d \times \mathcal{C}^{1,2}([0,T]\times\Dc) \rightarrow \mathcal{C}([0,T]\times\Dc)$ is given by 
\begin{equation*}
\begin{split}
    \mathcal{L}^a\varphi(t,x,\mathfrak{i},\mathfrak{v})
:=\frac12\,\mathrm{Tr}\Big(\sigma\sigma^\top(t,x,a)\,\frac{\partial^2 \varphi}{\partial x^2}(t,x,\mathfrak{i},\mathfrak{v})\Big)
+\Big\langle b(t,x,a), \frac{\partial \varphi}{\partial x}(t,x,\mathfrak{i},\mathfrak{v})\Big\rangle,
\end{split}
\end{equation*}
for $a\in\R$, and the execution operator $\mathcal{J}:  [0,T]\times\R^d \times \mathcal{C}^{1,2}([0,T]\times\Dc) \rightarrow \mathcal{C}([0,T]\times\Dc)$ is
\begin{equation*}
\resizebox{\textwidth}{!}{$
\begin{aligned}
\mathcal J v(t,x,\mathfrak i,\mathfrak v)
:= \sum_{i=1}^N
\mathds{1}_{\{\langle \mathfrak i,e_i\rangle>0\}}
\,\ell_i
\Big(
\varphi\big(t,\Gamma(t,x, \langle \mathfrak v,e_i\rangle),
\mathfrak i-\langle \mathfrak i,e_i\rangle e_i,
\mathfrak v-\langle \mathfrak v,e_i\rangle e_i\big)
- \varphi(t,x,\mathfrak i,\mathfrak v) - c(t,x,
\langle \mathfrak i,e_i\rangle ,
\langle \mathfrak v,e_i\rangle)
\Big).
\end{aligned}
$}\end{equation*}

Lastly, if $\sum_{i=1}^N\langle \mathfrak i,e_i\rangle=K$, the value function $v$ on $[0,T]\times\Dc$ is associated to 
\begin{equation}\label{HJB2}\begin{split}
    -\frac{\partial v}{\partial t}(t,x,\mathfrak i,\mathfrak v)&- \sup_{a\in\R}H^a\Big(t,x,\mathfrak i,\mathfrak v,v,\tfrac{\partial v}{\partial x},\tfrac{\partial^2 v}{\partial x^2}\Big) = 0.\end{split}
\end{equation}

In the equations above, we use the convention that if $m - \langle \mathfrak i, e_i \rangle$ is zero, then we drop the arguments in $\mathfrak i,\mathfrak v$.

\begin{remark}
    One could extend our framework and allow for marked point processes instead of the simple point processes we consider. That is, considering $(T^i_n,R^i_n)_{n\geq 1}$  instead of $(T^i_n)_{n\geq 1}$ for $i\in\mathcal{I}$.  In such a setup, similar to \cite{cartea2023optimal}, we could modulate the impulse operator with the mark $R^i_n$ associated with  a given execution time $T^i_n$. In such a case, the only variation from the equations above would be in the $\mathcal{J}$ operator, where one would have an expectation term weighting over the possible values of the marks.
\end{remark}

For any locally bounded function $v$ on $[0,T]\times \Dc$, we denote by $v^*$ and $v_*$
its upper- and lower-semicontinuous envelopes, defined for every 
$(t,x,\mathfrak i,\mathfrak v)\in [0,T]\times\Dc$ by
\begin{equation}\label{eq:envelopes_cont}
v^*(t,x,\mathfrak i,\mathfrak v)
:= \limsup_{\substack{(t',x')\to (t,x)}} 
v(t',x',\mathfrak i,\mathfrak v)
~~\text{and}~~
v_*(t,x,\mathfrak i,\mathfrak v)
:= \liminf_{\substack{(t',x')\to (t,x)}} 
v(t',x',\mathfrak i,\mathfrak v).
\end{equation}
It is clear that 
$v_* \le v \le v^*$, 
that $v^*$ (resp.~$v_*$) 
is upper (resp.~lower) semicontinuous, 
and that $v^* = v_* = v$ 
at all continuity points of $v$.

\begin{definition}[Viscosity solution]\label{viscosity_def}We say that a family of locally bounded functions $v$ define a viscosity supersolution (resp. subsolution) of \eqref{HJB1} and \eqref{HJB2} on $[0,T)\times\Dc$ if it satisfies:
   \begin{enumerate}
       \item For $( t, x, \mathfrak i,\mathfrak v) \in [0,T)\times\Dc$ and any smooth test function $\varphi \in C^{1,2}([0,T]\times\Dc)$ such that $(v_* - \varphi)$ attains a local minimum (resp. $(v^* - \varphi)$ attains a local maximum) at $(t, x, \mathfrak i,\mathfrak v)$ over the set $[t, t+\delta) \times B_\delta(x)\times \llbracket 1,K\rrbracket^K\times B_\delta(\mathfrak v) \subset [0,T)\times\Dc$ for some $\delta > 0$, we have 
\begin{align*}
\min \bigg\{&-\frac{\partial \varphi}{\partial t}(t,x,\mathfrak i,\mathfrak v)-\sup_{a\in\R}H^a\Big(t,x,\mathfrak i,\mathfrak v,\varphi,\tfrac{\partial \varphi}{\partial x},\tfrac{\partial^2 \varphi}{\partial x^2}\Big)\,,\, \big(\varphi - \mathcal{M}\varphi\big)(t,x,\mathfrak i,\mathfrak v)\bigg\}\geq 0\,(\text{resp.} \leq),
\end{align*}
 \item Additionally, for any $(t, x, \mathfrak i,\mathfrak v) \in [0,T)\times\Dc$ and any smooth test function $\varphi \in C^{1,2}([0,T]\times\Dc)$ such that $(v_* - \varphi)$ attains a local minimum (resp. $(v^* - \varphi)$ attains a local maximum) at $(t, x, \mathfrak i,\mathfrak v)$ over the set $[t, t+\delta) \times B_\delta(x)\times \llbracket 1,K\rrbracket^K\times B_\delta(\mathfrak v) \subset [0,T)\times\Dc$ for some $\delta > 0$, we have 
\begin{equation*}\begin{split}
    -\frac{\partial \varphi}{\partial t}(t,x,\mathfrak i,\mathfrak v)&- \sup_{a\in\R}H^a\Big(t,x,\mathfrak i,\mathfrak v,\varphi,\tfrac{\partial \varphi}{\partial x},\tfrac{\partial^2 \varphi}{\partial x^2}\Big) \geq 0\,(\text{resp.} \leq).\end{split}
\end{equation*}
\end{enumerate}
\vspace{-0.3cm}
The first part of the definition covers the case $\sum_{i=1}^N\langle\mathfrak i,e_i\rangle < K$ and the second the case $\sum_{i=1}^N\langle\mathfrak i,e_i\rangle = K$. A locally bounded function $v$ define a viscosity solution if it is both a viscosity supersolution and subsolution of \eqref{HJB1} and \eqref{HJB2}.
\end{definition}

\begin{proposition}\label{prop:viscosity_solution}
    The value function $v$ is a viscosity solution of \eqref{HJB1} and \eqref{HJB2}. Additionally,
    $$
v(T^-,x,\mathfrak i,\mathfrak v) =v(T,x,\mathfrak i,\mathfrak v)=g(x),\quad\forall(x,\mathfrak i,\mathfrak v)\in \Dc.
$$
\end{proposition}
\begin{proof}
    Refer to Appendix \ref{proof_viscosity}.
    
\end{proof}
\begin{theorem}
[Comparison principle]
\label{theo:comparaison} If $w$ is a viscosity subsolution of \eqref{HJB1} and \eqref{HJB2} and
$v$
is a viscosity supersolution of \eqref{HJB1} and \eqref{HJB2}, such that 
$$w^*(T, x,\mathfrak{i},\mathfrak{v})\leq
v_*(T, x,\mathfrak{i},\mathfrak{v}),$$
for all $(t,x,\mathfrak{i},\mathfrak{v})\in [0,T]\times\Dc$, then $w \leq v$ on $[0,T]\times\Dc$. Additionally, the unique viscosity solution of \eqref{HJB1} and \eqref{HJB2} is continuous on $[0,T]\times\Dc$ associated to the terminal condition $v(T,x,\mathfrak i,\mathfrak v)=g(x),$ for all $(x,\mathfrak i,\mathfrak v)\in \Dc$.
\end{theorem}
\begin{proof}
    Refer to Appendix \ref{proof_viscosity}.
    
\end{proof}
\subsection{On the Smooth-Fit Principle}
To apply the free-boundary regularity results of \cite{pham1997optimal,guo_smooth_fit} to our mixed control and impulse problem with execution delay, one needs to strengthen the assumptions in \ref{assumptions}. The viscosity framework developed earlier guarantees existence and uniqueness of the value function, but it does not in itself yield differentiability or regularity of the intervention boundary. The analysis of \cite{guo_smooth_fit} requires additional smoothness on the local dynamics, on the impulse mechanism, and on the cost structure so that the quasi-variational inequality admits a classical free-boundary interpretation. Recall that if $\sum_{i=1}^N\langle \mathfrak i,e_i\rangle<K$, then the following QVI is linked to the value function defined in \eqref{eq:value_function}:
$$
\min\Big\{
-\frac{\partial v}{\partial t}(t,x,\mathfrak i,\mathfrak v)
-\sup_{a\in\R}H^a\big(t,x,\mathfrak i,\mathfrak v,
v,\tfrac{\partial v}{\partial x},\tfrac{\partial^2 v}{\partial x^2}\big),\;
v(t,x,\mathfrak i,\mathfrak v)-\mathcal M v(t,x,\mathfrak i,\mathfrak v)
\Big\}=0.
$$

On the continuation region, the dynamics are governed purely by the continuous control $a$, which enters the PDE through a pointwise Hamiltonian maximization:
$$
-\tfrac{\partial v}{\partial t}(t,x,\mathfrak i,\mathfrak v)
= \sup_{a\in\R}
H^a\big(t,x,\mathfrak i,\mathfrak v,
v,\tfrac{\partial v}{\partial x},\tfrac{\partial^2 v}{\partial x^2}\big).
$$

The absolutely continuous control therefore acts only inside the continuation region
$$
\mathcal C^{\mathfrak i,\mathfrak v}
=\big\{(t,x): v(t,x,\mathfrak i,\mathfrak v)>\mathcal M v(t,x,\mathfrak i,\mathfrak v)\big\},
$$
and plays a local role, in contrast with impulses that generate nonlocal jumps. In particular, the continuous control affects neither the structure nor the smoothness of the impulse operator $\mathcal M$, but it determines the parabolic operator governing the value function inside the no-action region. If we further assume that the drift, diffusion, and running cost are $\mathcal C^{1+\beta}$ in the spatial variable (uniformly in $(t,a,\mathfrak i,\mathfrak v)$), and that the terminal reward is $\mathcal C^{2+\beta}$, then we would expect the nonlinear HJB operator
$$
F(t,x,\mathfrak i,\mathfrak v,r,q,p,X)
:=
-q-\sup_{a\in\mathcal A}H^a(t,x,\mathfrak i,\mathfrak v,r,p,X)
$$
to be uniformly parabolic with Hölder-continuous coefficients, as required in the interior Schauder theory. Likewise, if we assume that the impulse transition map $\Gamma(t,x,i,\xi)$ is locally Lipschitz in in $(t,x,\xi)$, and nondegenerate, and that the impulse cost $c(t,x,\xi,i)$ is strictly positive for $\xi\neq 0$, and satisfies a subadditivity condition preventing the accumulation of infinitely many impulses over finite time, then these assumptions ensure the regularity of the impulse operator. Moreover, if the minimization defining $\mathcal M v(t,x,\mathfrak i,\mathfrak v)$ admits a unique minimizer $\xi^\star(t,x,\mathfrak i,\mathfrak v)$ that is locally bounded in $(t,x)$ for every $(t,x,\mathfrak i,\mathfrak v)\in[0,T]\times\mathcal D$, then the structure of the impulse component aligns with the conditions imposed in \cite{guo_smooth_fit}. Under these strengthened assumptions, the following regularity result applies: the value function is $\mathcal C^1$ on the entire domain and $\mathcal C^2$ on the continuation region. In particular, the free boundary 
$$\partial\mathcal C^{\mathfrak i,\mathfrak v}:=\big\{(t,x): v(t,x,\mathfrak i,\mathfrak v)=\mathcal M v(t,x,\mathfrak i,\mathfrak v)\big\}$$
satisfies both value matching and smooth fit
$$
v=\mathcal M v~~\text{and}~~ D_x v = D_x (\mathcal M v).
$$
Moreover  the free boundary is $\mathcal C^1$ in $(t,x)$.
\begin{remark}
    We  expect the nonlocal integral term generated by these exogenous jumps to preserve the same interior regularity as the diffusion part.
\end{remark}

In the trading problem considered in Section \ref{sec:CEX-DEX} of this paper,  the optimal strategy for each pending-order configuration $(\mathfrak i,\mathfrak v)$ is characterized by a smooth free-boundary in $(t,x)$ that separates the region where the agent continues with the absolutely continuous control from the region where it becomes optimal to submit a DEX order with a given priority fee and size. Under the smooth-fit property, one would expect the marginal value with respect to $x$ to coincide on both sides of this boundary when transitioning from the pre- to the post-impulse state.

\section{CEX-DEX Optimal Trading Problem}\label{sec:CEX-DEX}
\subsection{Model Setup}
Using the setup above, we consider the case of $N$ priority fees and $K$ pending orders. For simplicity,  we include the base fee (which is a flat fee paid by all takers in the DEX) in the  ``priority fee'' paid by the controller.
Here, we consider the case of a two dimensional $\P$-Brownian motion $W = (W^Z, W^S)$. Let $t \in [0,T]$ and $y = (s,q,z,\mathfrak{i},\mathfrak{v}) \in \Dc$, where $\Dc$ defines the following domain
$$\Dc := \{(s,q,z,\mathfrak{i},\mathfrak{v}) : (s,q,z) \in \R_+^2 \times \R,\; 
\mathfrak{i} \in \mathbb{I},\; \mathfrak{v} \in \mathbb{V} \}.$$
For an execution strategy $\alpha \in \Ac_{K, \mathfrak{i},\mathfrak{v}}(t)$, 
the trader monitors the controlled state variable $X$ such that 
$$X^{\alpha} = (S^{\alpha}, Q^{\alpha}, Z^{\alpha}).$$
 The agent observes the external price in a centralised venue which we call $S$ and that follows 
\begin{equation}
     S^{t,s,\alpha}_u = s +  \sigma^S\,\big(W^S_u - W^S_t \big)\,,\quad u\in[t,T],
\end{equation}
where $s\in\R_+$ is the value of $S_t$.  The pool midprice, denoted by $Z$, evolves as 
\begin{equation}
   Z^{t,z,s,\alpha}_u = z +\int_t^u \kappa (S^{t,s,\alpha}_r - Z^{t,z,s,\alpha}_r)\,\d r  +  \sigma^Z\,\big(W^Z_u - W^Z_t \big) + \sum_{\tilde{\tau}_n\leq u} h(\xi_n, Z^{t,z,s,\alpha}_{\tilde{\tau}^-_n})\,\,,\quad u\in[t,T],
\end{equation}
where $z\in\R_+$ is the value of $Z_t$. In what follows, we drop the starting values of the controlled processes for ease of the notation.  Here, the function $h$ is given by $h(\xi,Z) = \psi(Z)\,\xi$, and the function $\psi$ (which approximates the transaction price \cite{cartea2025decentralised}) is  $\psi(Z) = 2\,Z^{3/2}/d$ where $d$ is the depth of the DEX. Here, the drift term $\int_t^u \kappa\,(S_r^\alpha-Z_r^\alpha)\,\d r$ models mean reversion of the DEX state toward the CEX price $S^\alpha$: the parameter $\kappa>0$ quantifies the speed at which arbitrageurs realign prices across venues. The Brownian component $\int_t^u \sigma^Z\,\d W_r^Z$ captures exogenous/noise-trader activity. Finally, the jump term $\sum_{\tilde{\tau}_n\le u} h(\xi_n, Z^\alpha_{\tilde{\tau}_n^-})$ accounts for the instantaneous impact of the agent's discrete DEX trades executed at times $\tilde{\tau}_n$. The impact is linear in $\xi$, increases with the current level $Z$, and is attenuated by the constant-product pool depth $d$. The inventory of the agent is given by
\begin{equation}
    Q^{\alpha}_u = q + \underbrace{ \int_t^u \nu_r\,\d r }_{\text{trading in CEX}} + 
\underbrace{\sum_{n:\,\tilde{\tau}_n \leq u} \xi_n}_{\text{trading in DEX}}\,\,,\quad u\in[t,T],
\end{equation}
where $q\in\R$ is the inventory of the agent at time $t$,  and the cash accumulated from time $t$  to time $u$ is
\begin{equation}
    C^{\alpha}_u =   \underbrace{ - \int_t^u \nu_r\,\left(S^{\alpha}_r + k\,\nu_r\right)\d r}_{\text{trading in CEX}} + \underbrace{ \sum_{n:\,\tilde{\tau}_n \leq u} \big(\gamma(\xi_n, Z^{\alpha}_{\tilde{\tau}^-_i}) - \mathfrak{p}_{I_n}\big)}_{\text{trading in DEX}} \,\,,\quad u\in[t,T],
\end{equation}
where 
$$\gamma(\xi,z) = \frac{d^2}{\xi-d\,\sqrt{1/z}} + d\,\sqrt{z}\,.$$
The formula for $\gamma$ is the closed-form expression for the execution cost of a trade of size $\xi_n$ in the DEX.\\

The elements $(\mathfrak{i},\mathfrak{v}) = (0_{\mathbb{I}},0_{\mathbb{V}})$ denote the absence of pending orders. For $y=(s,q,z,\mathfrak{i}, \mathfrak{v})$ and for control $\nu$, we consider the following  running reward $f$ and terminal payoff $g$
\begin{align*}
  f(t,y,\nu) 
  &:= - \nu\,\left(s + k\,\nu\right) -\phi \,q^2~~\text{and}~~ 
  g(y):= q\,s - \Xi\,q^2,
\end{align*}
with $k>0$ and $\phi,\Xi\geq 0$. 
Lastly, the intervention cost $c$ is 
\begin{equation}
  c(\tilde{\tau}_n,Z^{\alpha}_{\tilde{\tau}^-_n}, \xi_n,I_n)\;:=\;  \gamma(\xi_n, Z^{\alpha}_{\tilde{\tau}^-_n}) - \mathfrak{p}_{I_n}\,,
\end{equation}
where for simplicity we exclude base fees as these could absorbed in the priority fees vector $\mathfrak{p}$.\footnote{The framework is flexible enough to accommodate strategic fees \cite{baggiani2025optimal} or alternative designs \cite{cartea2024strategic}. More precisely, one can add here any deterministic map that adjusts costs depending on the transaction size $\xi_n$ or the state of the pool. We leave these generalizations out for simplicity.  }
For any admissible control $\alpha \in \Ac_{K, \mathfrak i, \mathfrak v}(t)$, we define the performance criterion
Formally, for any admissible control $\alpha \in \Ac_{K, \mathfrak i, \mathfrak v}(t)$, we define the performance criterion
\begin{equation}
\resizebox{\textwidth}{!}{$
\begin{aligned}
  J(t,y,\alpha) = 
  \E_{t,y}\Bigg[
     - \int_t^T \nu_r\,\left(S^{\alpha}_r + k\,\nu_r\right)\d r  \,+ Q^{\alpha}_T\,S^{\alpha}_T - \Xi\,\big(Q^{\alpha}_T\big)^2- \sum_{\tilde \tau_n \in (t,T]} 
c(\tilde{\tau}_n,Z^{\alpha}_{\tilde{\tau}_n}, \xi_n,I_n) - \phi\,\int_t^T\big(Q^{\alpha}_r\big)^2\,\d r
    \Bigg].
\end{aligned}
$}\end{equation}
The associated value function is then given by
\begin{equation}\label{eq:value_function_cex_dex}
    v(t,y) := \sup_{\alpha \in \Ac_{K, \mathfrak i, \mathfrak v}(t)} J(t,y,\alpha).
\end{equation}
In particular, when there are no pending orders, the state reduces to 
$y_0 = (s,q,z,0_{\mathbb{I}},0_{\mathbb{V}}) \in \Dc$, and the value function defined in \eqref{eq:value_function} coincides with the one obtained under  $\mathcal{A}_{K,0,0}(t)$. 
\begin{nota}
     The conditional expectation given $(S^{\alpha}_t=s,Q^{\alpha}_t=q,Z^{\alpha}_t=z)$ under the probability measure $\mathbb{P}$ is denoted by $$\mathbb{E}_{t,y}[\cdot] = \mathbb{E}^{\mathbb{P}}\Big[\cdot\Big| S^{\alpha}_t =s,Q^{\alpha}_{t^-}=q,Z^{\alpha}_t= z\Big].$$
     
\end{nota}

Next, we characterize the  HJB equation associated with $v(t,y)$. If $\sum_{i=1}^N\langle \mathfrak i,e_i\rangle<K$, then the QVI that is linked to the value function is
\begin{align}\label{HJB1:CEXDEX}
\min \bigg\{&-\frac{\partial v}{\partial t}(t,y) -\sup_{\nu\in\R}H^\nu\Big(t,y,v,\tfrac{\partial v}{\partial x},\tfrac{\partial^2 v}{\partial x^2}\Big)\,,\, \big(v - \mathcal{M}v\big)(t,y)\bigg\}=0,
\end{align}
where $\mathcal{M}:\mathcal{C}([0,T]\times\Dc) \rightarrow \mathcal{C}([0,T]\times\Dc)$ is 
\begin{equation}\label{CEX-DEX intervention_operator}
\mathcal{M}\varphi(t,y)
:= \sup_{(\xi,i)\in[-\bar{V},\bar{V}]\times\setK}
\varphi(t,s,q,z, \mathfrak i + e_i , \mathfrak v + \xi e_i),
\end{equation}
and the Hamiltonian $H^\nu : [0,T]\times\Dc \times \mathcal{C}([0,T]\times\Dc)^3 \rightarrow \mathbb{R}$ is
\begin{align*}
    &H^\nu\Big(t,y,\varphi,\tfrac{\partial \varphi}{\partial x},\tfrac{\partial^2 \varphi}{\partial x^2}\Big)= 
 \mathcal{L}^\nu \varphi(t,y)+ \mathcal J \varphi(t,y)
-\phi\,q^2 - \nu(s+k\,\nu),
\end{align*}
with infinitesimal generator
\begin{equation*}
\begin{split}
\mathcal{L}^\nu\varphi(t,y)
:= \kappa(s-z)\,\frac{\partial v}{\partial z}(t,y) + \frac{1}{2}(\sigma^Z)^2\frac{\partial^2 v}{\partial z^2}(t,y)  + \frac{1}{2}(\sigma^S)^2\frac{\partial^2 v}{\partial s^2}(t,y) +\nu\,\frac{\partial v}{\partial q}(t,y) ,
\end{split}
\end{equation*}
execution operator of pending orders 
$$
\mathcal J v(t,y)
:= \sum_{i=1}^N
\mathds{1}_{\{\mathfrak i_i>0\}}
\,\ell_i
\Big(
\varphi\big(t,\Gamma(t,s,q,z, \xi_i),
\mathfrak i- \mathfrak i_i e_i,
\mathfrak v-\xi_i e_i\big)
- \varphi(t,y) + c(t,s,q,z, \mathfrak i_i ,
\xi_i)
\Big),
$$
where $\xi_i:=\langle\mathfrak{v},e_i\rangle$, $\mathfrak{i}_i:=\langle\mathfrak{i},e_i\rangle $ and the impulse operator $\Gamma$ is
$$
\Gamma(t,s,q,z,\xi) = (s, q +\xi, z + h(\xi,z) ).
$$

If $\sum_{i=1}^N\langle \mathfrak i,e_i\rangle=K$, the value function $v$ on $[0,T]\times\Dc$ is associated to 
\begin{equation}\label{HJB2:CEXDEX}
\begin{split}
    -\frac{\partial v}{\partial t}(t,y)&- \sup_{\nu\in\R}H^\nu\Big(t,y,v,\tfrac{\partial v}{\partial x},\tfrac{\partial^2 v}{\partial x^2}\Big) = 0.\end{split}
\end{equation}

Lastly, from the first order optimality condition we have $\nu^* = (2\,k)^{-1}(\frac{\partial v}{\partial q}  - s)$.

\subsection{Numerical Results}
We employ the values of model parameters in \cite{aqsha2025equilibrium} which are calibrated to ETH-USDC market data.\footnote{See \cite{di2025deviations} for a recent article studying the stylized facts of prices, liquidity, and order flow in DEXs. } More precisely, we take $S_0 = Z_0 = 2820$, $\sigma^S= 0.0569 \times S_0$, $\kappa = 1$, $\sigma^Z = 0.00569 \times S_0$. The depth is $d = 50,000$ ETH. The agent-specific aversion parameters are $\Xi=1$ and $\phi = 1$. 
Lastly, to obtain the value function we  use standard numerical schemes. 
For the fees and delays vectors, we consider the following. For experiments having $N\geq 1$ priority fees, we take the fee vector to be $\mathfrak p=\{\mathfrak p_1,p_1+200,\dots,\mathfrak p_1+200\,N\}$ with $\mathfrak p_1 = 100$,  and the delay vector is $\mathfrak{l}=\{\ell_1,\ell_1 + 0.5, \dots,\ell_1 +0.5\,N\}$ and $\ell_1 = 2$.

\paragraph*{Optimal trading strategy} Figure \ref{fig:value_q0} shows the value function and the corresponding exercise region for when the inventory is zero ($q=0$).\\ 
\begin{figure}[H]
    \centering
    \begin{subfigure}[t]{0.3\textwidth}
        \centering
        \includegraphics[width=\textwidth]{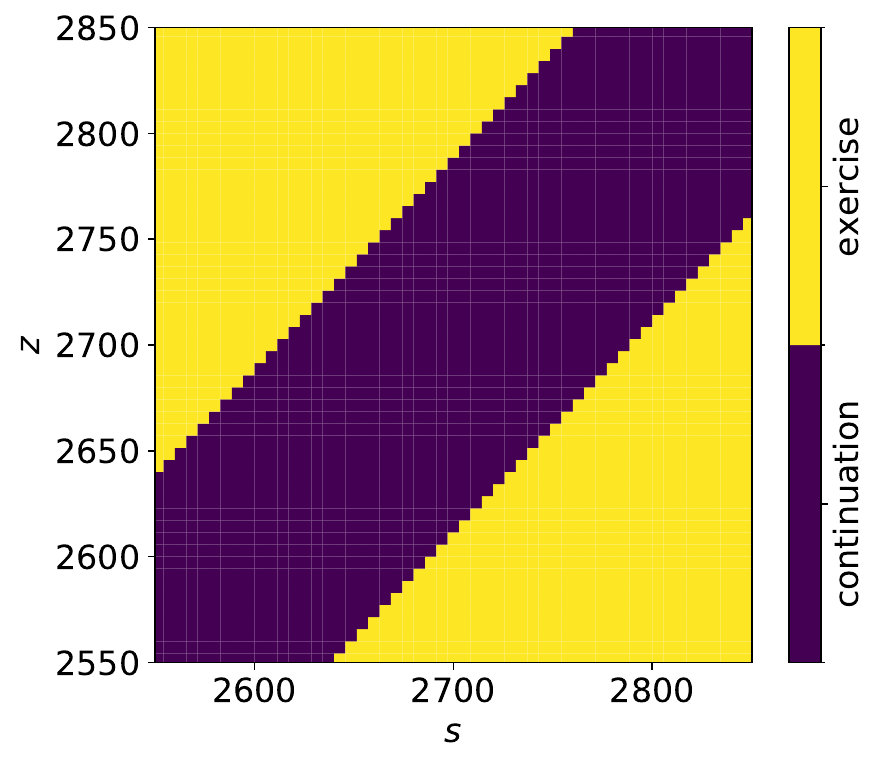}
    \end{subfigure}\hspace{0.5cm}
\begin{subfigure}[t]{0.32\textwidth}
        \centering
        \includegraphics[width=\textwidth]{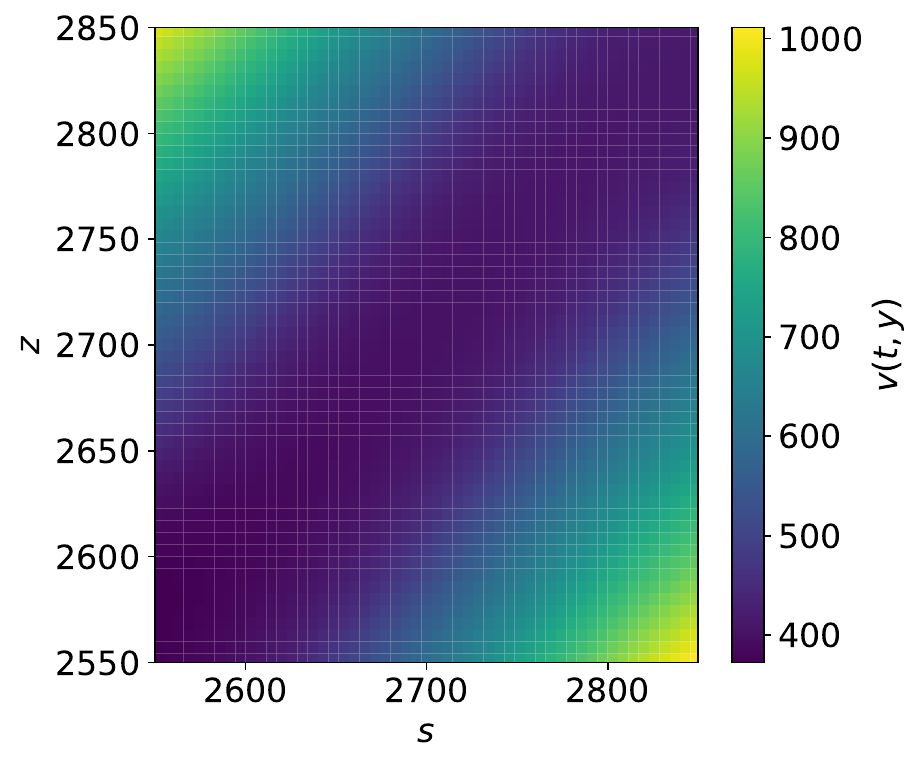}
 \end{subfigure}
       \caption{Value function and exercise region as a function of the CEX and DEX prices $(s,z)$ for $q=0$.}
       \label{fig:value_q0}
\end{figure}
The continuation region in Figure \ref{fig:value_q0} forms a diagonal band around $s=z$, where prices are closely aligned and it is optimal to wait, while the exercise region lies outside this band and expands as the divergence between the two prices increases. This is similar to the exit region found in \cite{bergault2025optimal} for when LPs find it optimal to exit the pool; in the present context, the region modulates the trading activity of the liquidity taker. This symmetry highlights that execution becomes optimal precisely when the price discrepancy between the two venues is sufficiently large, and not on the price direction. This reflects an arbitrage-driven incentive to trade whenever CEX and DEX prices dislocate.\\

Figure \ref{fig:trading_rate} illustrates the sensitivity of the optimal trading rate $\nu^*$ to the CEX price and the inventory ($z = 2700$).\\
 \begin{figure}[H]
        \centering
        \includegraphics[width=0.33\textwidth]{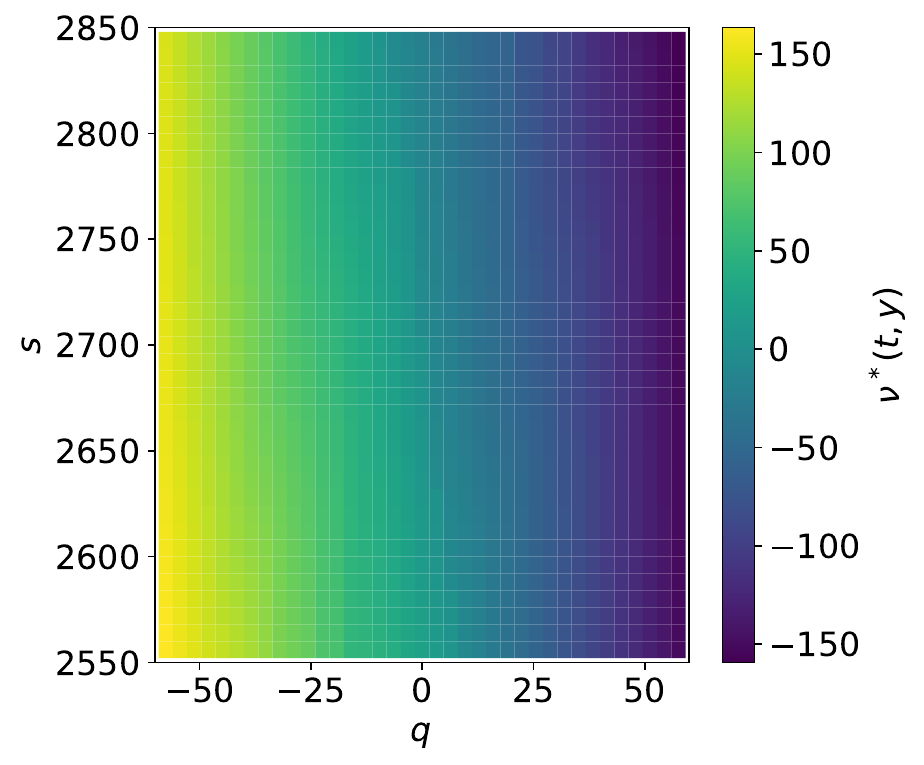}
       \caption{Optimal trading rate $\nu^*$ as a function of the CEX price $s$ and the inventory $q$ for $z=2700$ and $t=0.1$.}
       \label{fig:trading_rate}
\end{figure}
The plot indicates that the optimal trading rate $\nu^*$ depends primarily on the inventory, with only a weak sensitivity to the CEX price level. This is consistent with the fact that, under our specification, CEX prices are martingales, so there is no directional expected return to exploit through $s$. As a result, the trading rate is mainly driven by the inventory objective and its associated constraint.\\

\paragraph*{Inventory}
Figure \ref{fig:priority_fees_inventory} explores the effect of the inventory in  the priority fees. We take the inventory to be large and positive $q=20$ (negative case is analogous).\\ 
\begin{figure}[H]
    \centering
    \begin{subfigure}[t]{0.3\textwidth}
        \centering
        \includegraphics[width=\textwidth]{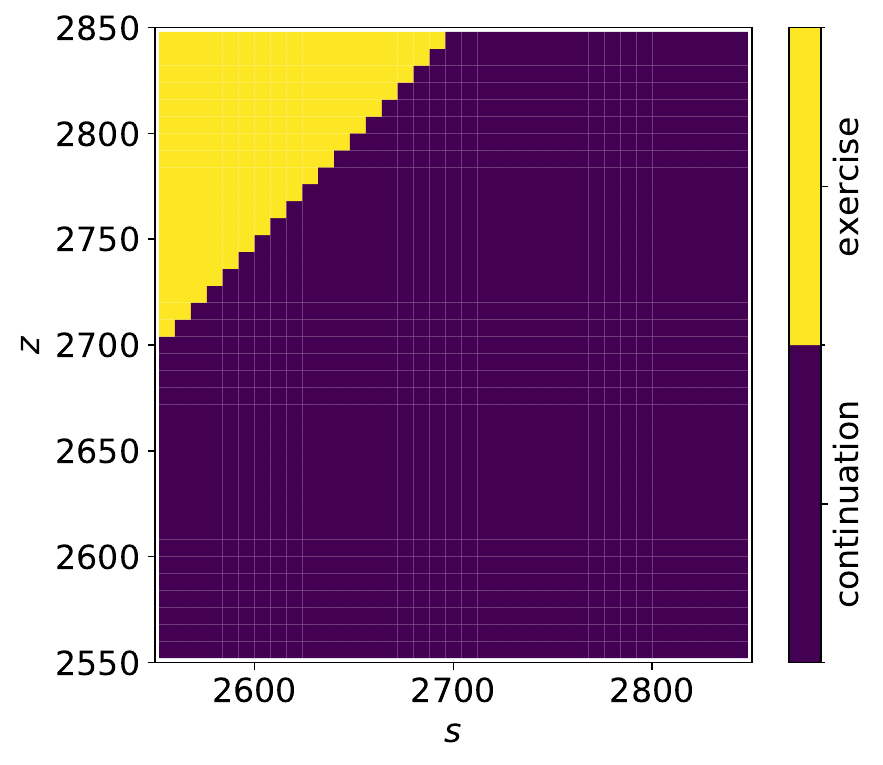}
        \caption{$t=t_0$}
        \label{fig:inv_t0}
    \end{subfigure}\hspace{0.3cm}
    \begin{subfigure}[t]{0.3\textwidth}
        \centering
        \includegraphics[width=\textwidth]{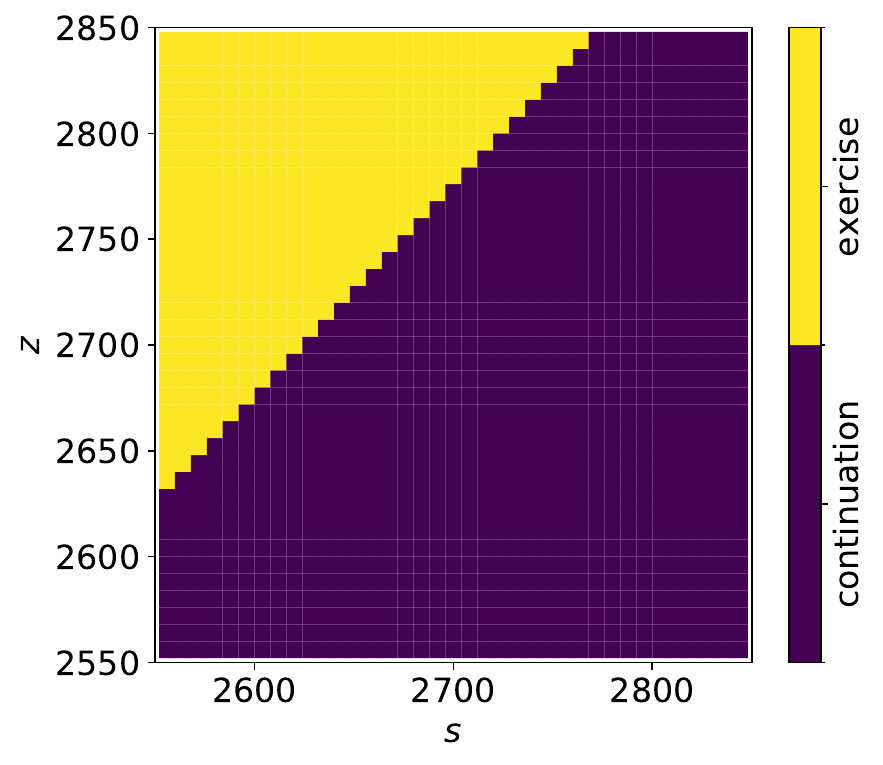}
        \caption{$t=t_1$}
        \label{fig:inv_t1}
    \end{subfigure}\hspace{0.3cm}
\begin{subfigure}[t]{0.3\textwidth}
        \centering
        \includegraphics[width=\textwidth]{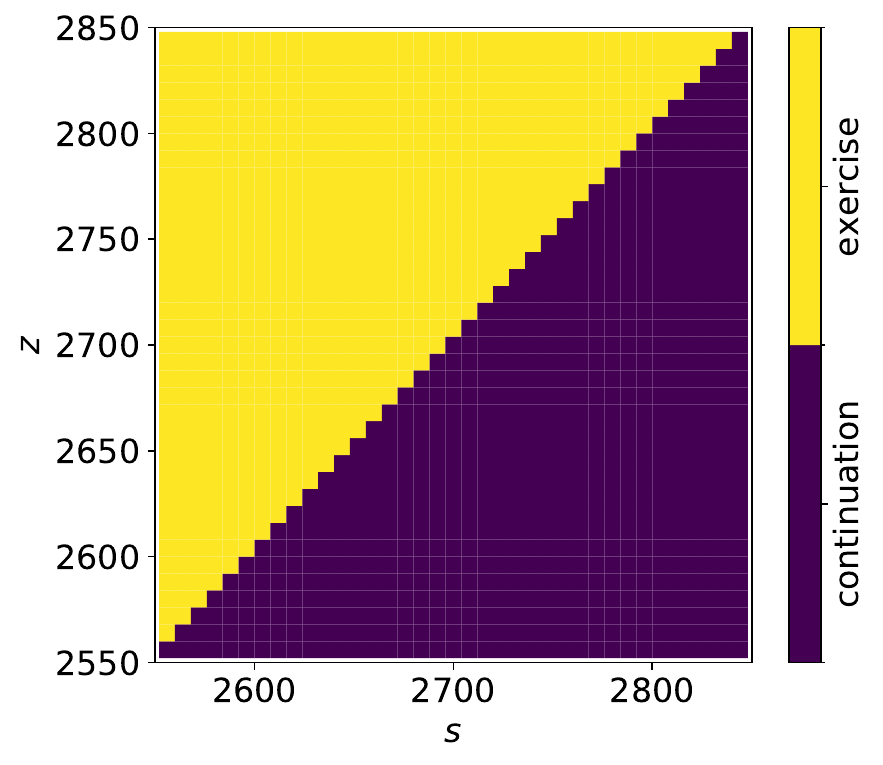}
        \caption{$t=t_2$}
        \label{fig:inv_t2}
 \end{subfigure}
    \caption{Exercise and continuation regions as a function of the CEX and DEX prices $(s,z)$ at time indices $t_0<t_1<t_2$ for $q=20$. Here, $t_0=0.2$, $t_1=0.5$, and $t_2 = 0.8$.}
    \label{fig:priority_fees_inventory}
\end{figure}
We see that the exercise region becomes one-sided, this is because the trader does not wish to acquire more inventory and the main incentive is to bring $q$ closer to zero towards the end of the time horizon $T$. As time passes (from left panel to right panel), the urgency to bring inventory closer to zero shifts the exercise region to be more and more aggressive.

\paragraph*{Priority fees} Figure~\ref{fig:priority_fees_time} illustrates the optimal priority fee policy. As the terminal time is approached, the continuation region progressively shrinks, while the exercise regions associated with positive priority fees expand. As one would expect, higher priority fees are reserved for when dislocations are larger. This is because the trader wishes to ensure that the arbitrage opportunity is attainable.\\
\begin{figure}[H]
    \centering
    \begin{subfigure}[t]{0.3\textwidth}
        \centering
        \includegraphics[width=\textwidth]{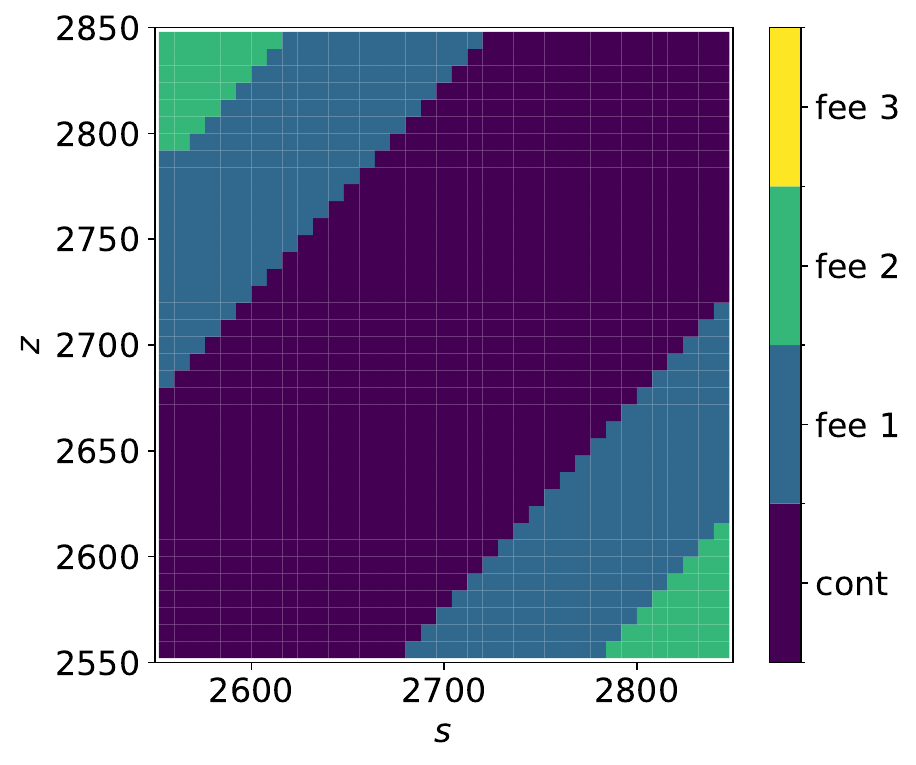}
        \caption{$t=t_0$}
        \label{fig:pf_t0}
    \end{subfigure}\hspace{0.2cm}
    \begin{subfigure}[t]{0.3\textwidth}
        \centering
        \includegraphics[width=\textwidth]{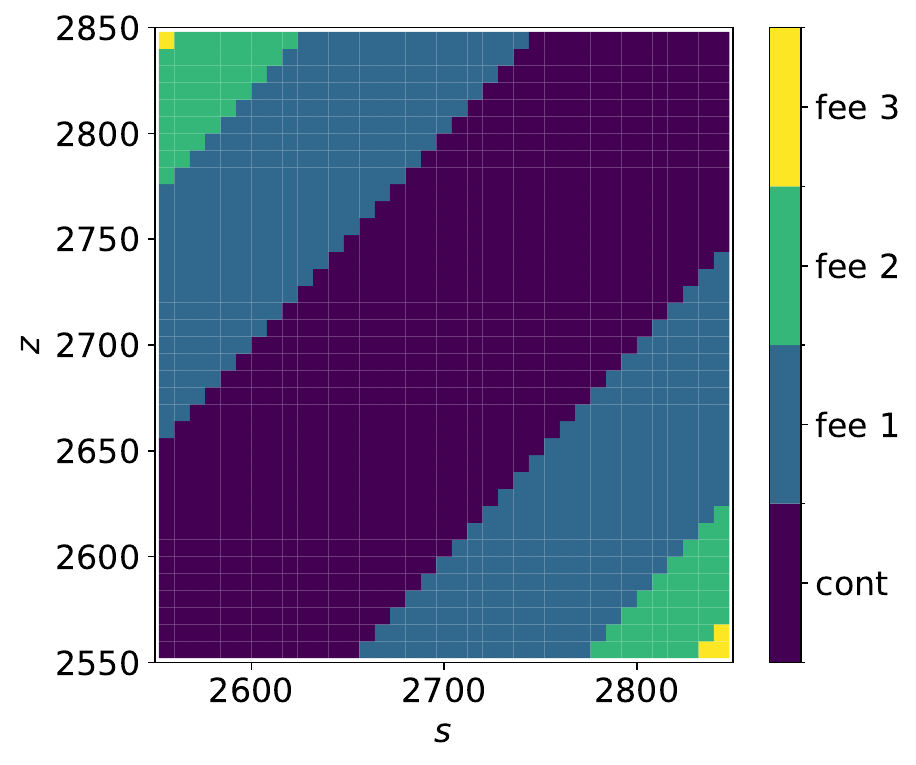}
        \caption{$t=t_1$}
        \label{fig:pf_t1}
    \end{subfigure}\hspace{0.2cm}
\begin{subfigure}[t]{0.3\textwidth}
        \centering
        \includegraphics[width=\textwidth]{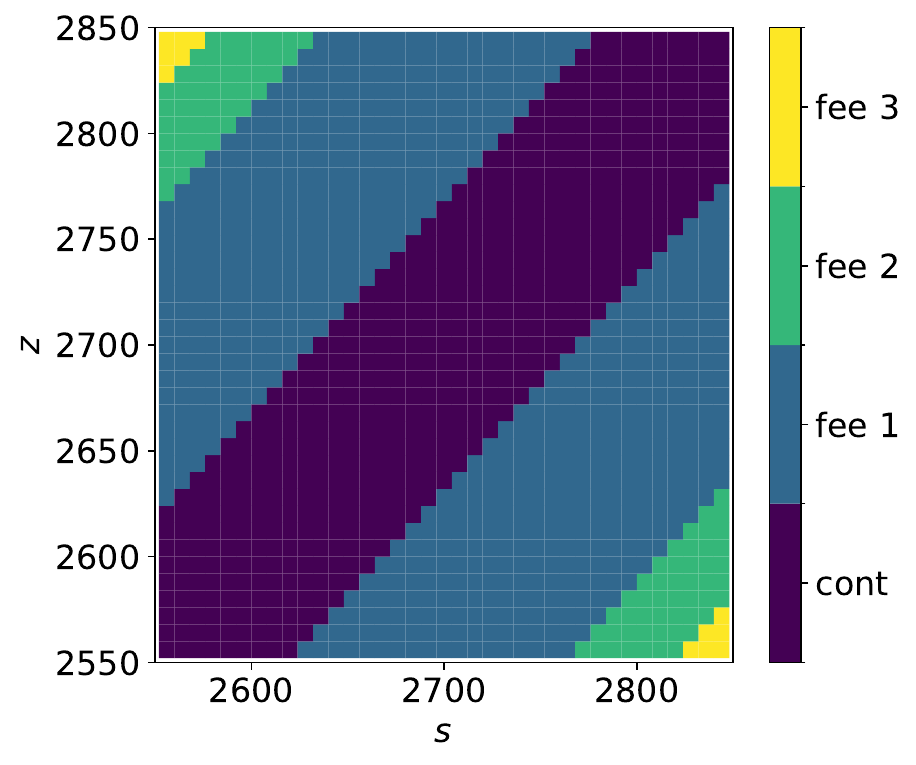}
        \caption{$t=t_2$}
        \label{fig:pf_t2}
 \end{subfigure}
    \caption{Priority fees as a function of the CEX and DEX prices $(s,z)$ at time indices $t_0<t_1<t_2$ for $q=0$. Here, $t_0=0.2$, $t_1=0.5$, and $t_2 = 0.8$.}
    \label{fig:priority_fees_time}
\end{figure}
The plots also reflect the increasing urgency to liquidate positions as time passes: states for which it was optimal to wait at earlier times are gradually replaced by decisions to pay for execution priority. In particular, higher priority fee levels become optimal closer to the terminal condition, indicating that the agent is willing to incur larger execution costs in order to reduce liquidation risk.\\

Figure \ref{fig:priority_fees_intensities} shows the effect of increasing exogenous arrival intensities on the optimal priority fee policy.\\
\begin{figure}[H]
    \centering
    \begin{subfigure}[t]{0.3\textwidth}
        \centering
        \includegraphics[width=\textwidth]{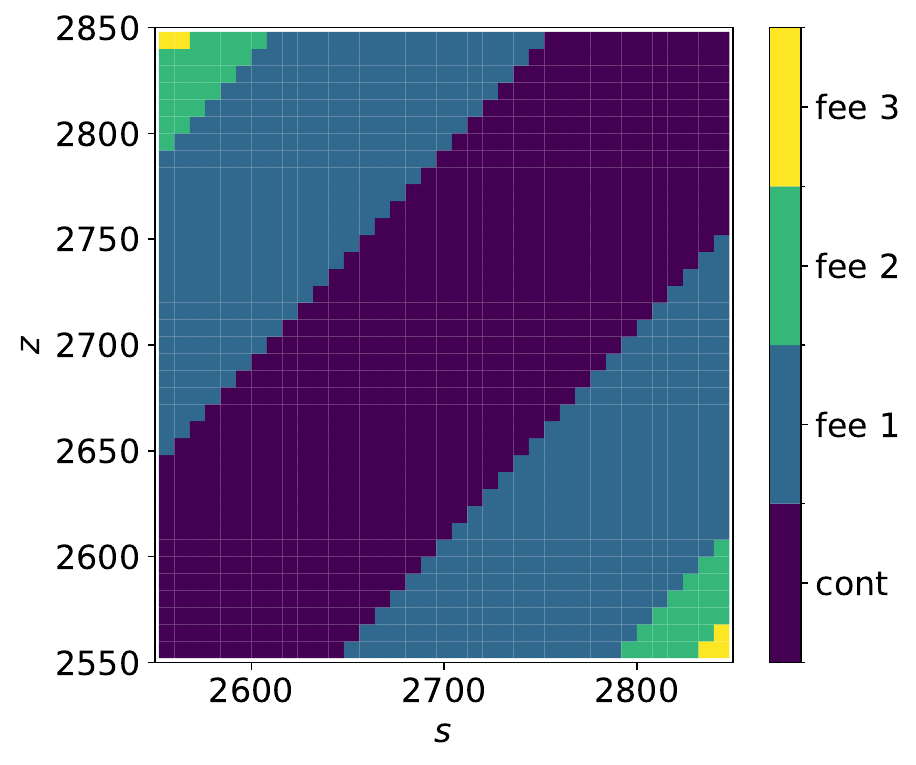}
        \caption{$\mathfrak{l}=\mathfrak{l}_0$}
        \label{fig:pf_l0}
    \end{subfigure}\hspace{0.2cm}
    \begin{subfigure}[t]{0.3\textwidth}
        \centering
        \includegraphics[width=\textwidth]{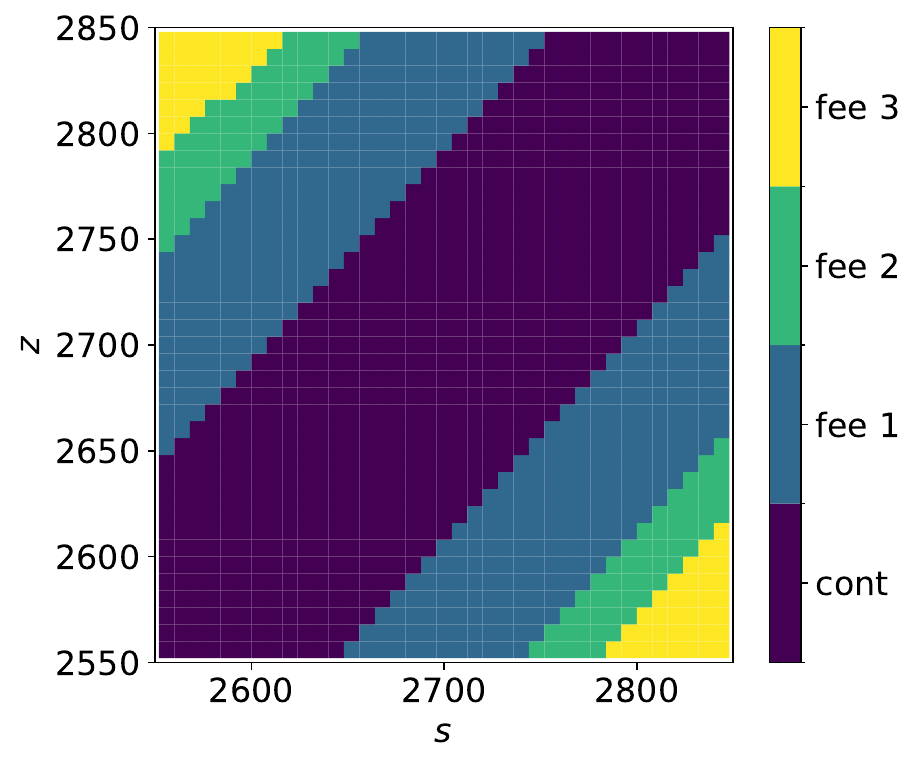}
        \caption{$\mathfrak{l}=\mathfrak{l}_1$}
        \label{fig:pf_l1}
    \end{subfigure}\hspace{0.2cm}
\begin{subfigure}[t]{0.3\textwidth}
        \centering
        \includegraphics[width=\textwidth]{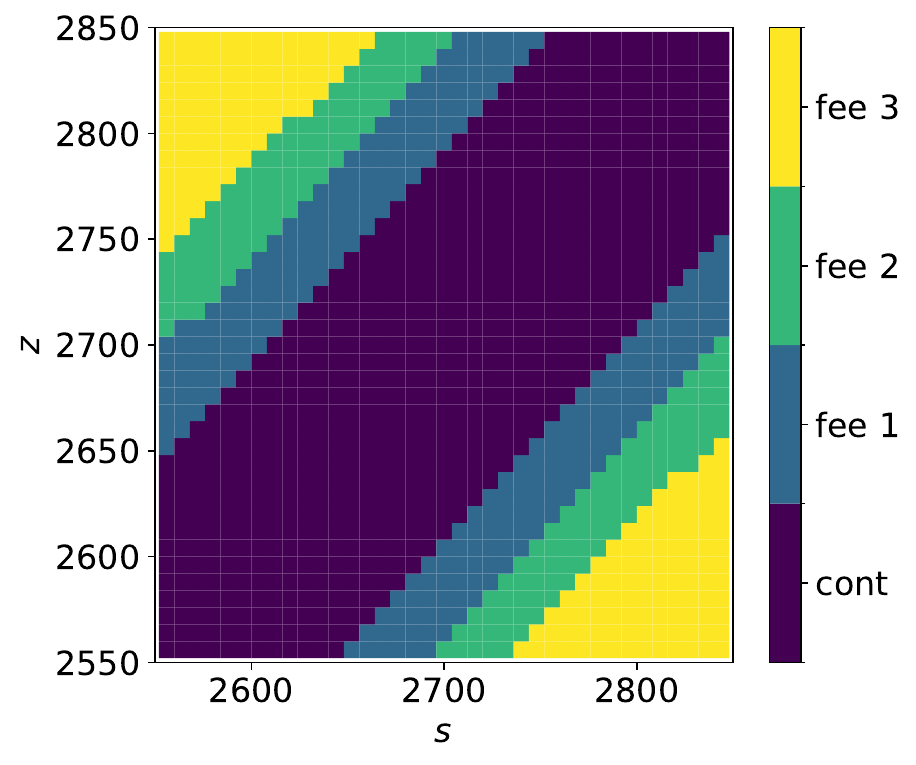}
        \caption{$\mathfrak{l}=\mathfrak{l}_2$}
        \label{fig:pf_l2}
 \end{subfigure}
    \caption{Priority fee regions for increasing exogenous arrival intensities at the CEX and DEX at time $t = 0.5$. From left to right, all arrival intensities are jointly increased $\mathfrak{l}_2<\mathfrak{l}_1<\mathfrak{l}_0$.}
    \label{fig:priority_fees_intensities}
\end{figure}
 As arrival intensities increase, and hence expected execution delays decrease, the regions associated with lower priority fees expand, while those corresponding to higher priority fees shrink. Economically, this means that when execution becomes faster even at low priority levels, the marginal benefit of paying for additional priority is reduced. Consequently, the agent optimally relies more often on lower priority fees, since similar execution speed can be achieved at a lower cost.\\

In the next experiment we study the additional benefits the controller gets when entertaining more priority fees in their control problem. We start from a single priority level and then increase the number of available priority levels to $N\in\{1,2,3,5,7,10,30,50\}$. For each additional level, we add one more fee-intensity pair, where both the priority fee and its corresponding execution intensity are increased by a fixed increment relative to the previous level (see the start of this section for details).
Figure \ref{tab:v0_norm_sweep} reports the norm of the value function $v_0: (s,q,z, \mathfrak i, \mathfrak v)\mapsto v(0,s,q,z, \mathfrak i, \mathfrak v)$ at time $t=0$ computed on the full $(s,q,z)$ grid (with $\mathfrak i = 0_{\mathbb{I}}$ and $\mathfrak v =0_{\mathbb{V}}$). \\

\begin{figure}[H]
        \centering
        \includegraphics[width=0.5\textwidth]{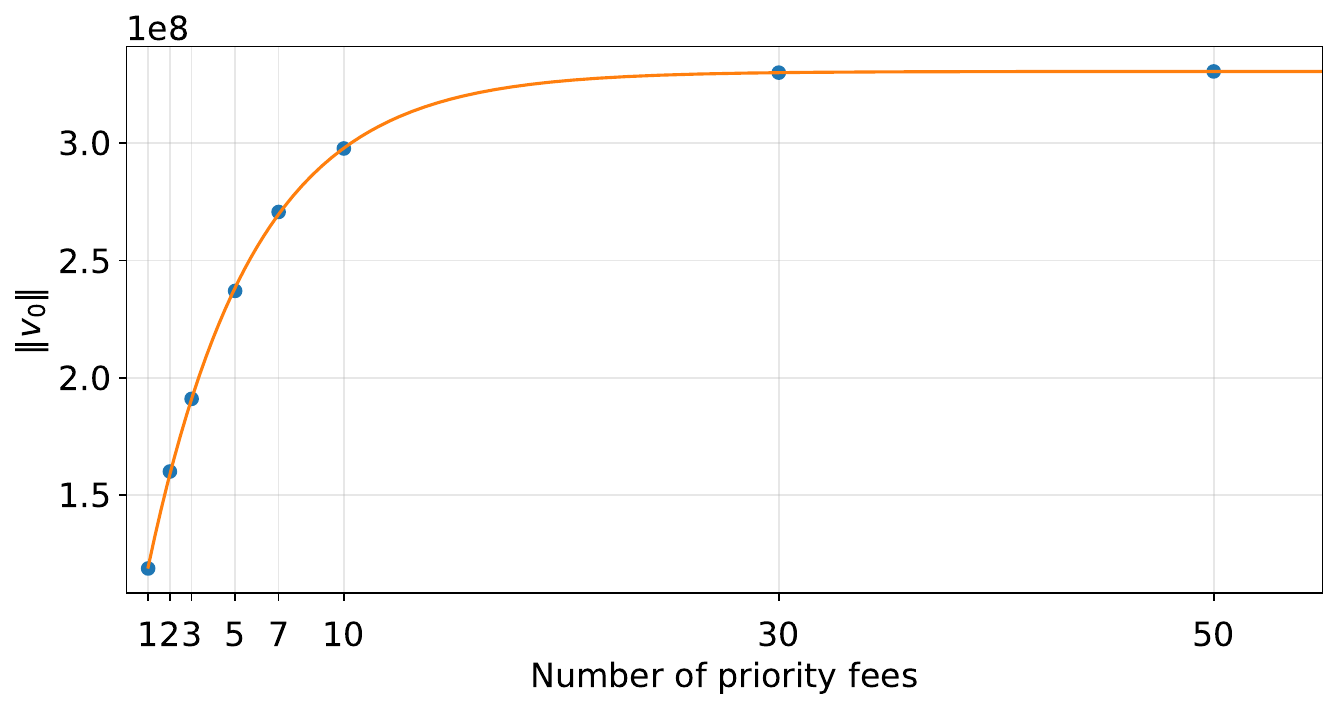}
       \caption{$\|v_0\|$ over all states as a function of the number of priority fee levels.}
\label{tab:v0_norm_sweep}
\end{figure}

The figure shows that, as the number of priority-fee options increases, the value function improves relative to the benchmark with a single priority fee, which we interpret as higher attainable profit from having finer control over execution speed and latency risk. We see that the outperformance rapidly plateaus, in fact, beyond thirty, the outperformance of employing additional priority fees becomes negligible. 
\\

Lastly, we compare the performance of the optimal priority-fee policy against that from a randomized baseline within the same mixed-control QVI. In both cases, we use the continuous CEX control $\nu^*$ and the same intervention times. The only modification is the fee choice upon intervention: the optimal policy selects the maximizing priority level, while the randomized baseline draws the fee index $J\sim\mathrm{Unif}\{1,\dots,N\}$ (with the remaining impulse decision rule unchanged), with $N = 3$. We estimate the randomized strategy by averaging over 100 independent samples. Using $\|v_0\|$ over all states as the metric, the optimal policy improves the value function by $18.2\%$ relative to the random-fee baseline.

\section{Conclusions}
We introduced a new type of mixed control problem where the agent is allowed 
to choose the expected delay of the execution of their impulses. The associated value function to this optimization problem is shown to satisfy a system of variational Hamilton–Jacobi–Bellman inequalities in the viscosity sense. Moreover, we establish uniqueness of the solution to this HJB-QVI. We apply our formulation to the problem of optimal trading between CEX and DEX, where the agent chooses the priority fee attached to the orders sent to the DEX. We find the optimal priority fee as a function of time, inventories, and price dislocations between the CEX and the DEX. Our results show that the outperformance one gets from employing more priority fees rapidly plateaus, and that the optimal fee selection brings about significant outperformance over a non-strategic fee selection.

\section*{Comments}
For the purpose of open access, the authors have applied a CC BY public copyright licence to any author accepted manuscript version arising from this submission.

\bibliographystyle{plain}
\bibliography{references.bib}

@article{cartea2023optimal,
  title={Optimal execution with stochastic delay},
  author={Cartea, {\'A}lvaro and S{\'a}nchez-Betancourt, Leandro},
  journal={Finance and Stochastics},
  volume={27},
  number={1},
  pages={1--47},
  year={2023},
  publisher={Springer}
}

@article{capponi2026price,
  title={Price discovery on decentralized exchanges},
  author={Capponi, Agostino and Jia, Ruizhe and Yu, Shihao},
  journal={The Review of Financial Studies},
  pages={hhag002},
  year={2026},
  publisher={Oxford University Press}
}

@article{campbell2025optimal,
  title={Optimal fees for liquidity provision in automated market makers},
  author={Campbell, Steven and Bergault, Philippe and Milionis, Jason and Nutz, Marcel},
  journal={arXiv preprint arXiv:2508.08152},
  year={2025}
}

@article{capponi2025longer,
  title={Do Longer Block Times Impair Market Efficiency in Decentralized Markets?},
  author={Capponi, Agostino and Cartea, Alvaro and Drissi, Fay{\c{c}}al},
  journal={Available at SSRN 5290232},
  year={2025}
}

@article{hasbrouck2022need,
  title={The need for fees at a dex: How increases in fees can increase dex trading volume},
  author={Hasbrouck, Joel and Rivera, Thomas J and Saleh, Fahad},
  journal={Available at SSRN},
  volume={4192925},
  year={2022}
}

@article{bergault2025optimal,
  title={Optimal exit time for liquidity providers in automated market makers},
  author={Bergault, Philippe and Bieber, S{\'e}bastien and S{\'a}nchez-Betancourt, Leandro},
  journal={arXiv preprint arXiv:2509.06510},
  year={2025}
}

@article{cartea2025decentralised,
title = {Decentralised finance and automated market making: Execution and speculation},
journal = {Journal of Economic Dynamics and Control},
volume = {177},
pages = {105134},
year = {2025},
issn = {0165-1889},
doi = {https://doi.org/10.1016/j.jedc.2025.105134},
url = {https://www.sciencedirect.com/science/article/pii/S0165188925001009},
author = {\'Alvaro Cartea and Fayçal Drissi and Marcello Monga}
}

@article{di2025deviations,
  title={Deviations from Tradition: Stylized Facts in the Era of DeFi},
  author={Di Nosse, Daniele Maria and Gatta, Federico and Lillo, Fabrizio and Jaimungal, Sebastian},
  journal={arXiv preprint arXiv:2510.22834},
  year={2025}
}

@article{fukasawa2024model,
  title={Model-Free Hedging of Impermanent Loss in Geometric Mean Market Makers with Proportional Transaction Fees},
  author={Fukasawa, Masaaki and Maire, Basile and Wunsch, Marcus},
  journal={Applied Mathematical Finance},
  volume={31},
  number={2},
  pages={108--129},
  year={2024},
  publisher={Taylor \& Francis}
}

@article{he2025arbitrage,
  title={Arbitrage on decentralized exchanges},
  author={He, Xue Dong and Yang, Chen and Zhou, Yutian},
  journal={arXiv preprint arXiv:2507.08302},
  year={2025}
}

@inproceedings{cartea2023execution,
  title={Execution and statistical arbitrage with signals in multiple automated market makers},
  author={Cartea, {\'A}lvaro and Drissi, Fay{\c{c}}al and Monga, Marcello},
  booktitle={2023 IEEE 43rd International Conference on Distributed Computing Systems Workshops (ICDCSW)},
  pages={37--42},
  year={2023},
  organization={IEEE}
}

@article{baggiani2025optimal,
  title={Optimal Dynamic Fees in Automated Market Makers},
  author={Baggiani, Leonardo and Herdegen, Martin and S{\'a}nchez-Betancourt, Leandro},
  journal={arXiv preprint arXiv:2506.02869},
  year={2025}
}

@article{cartea2024strategic,
  title={Strategic bonding curves in automated market makers},
  author={Cartea, {\'A}lvaro and Drissi, Fay{\c{c}}al and S{\'a}nchez-Betancourt, Leandro and Siska, David and Szpruch, Lukasz},
  journal={Available at SSRN 5018420},
  year={2024}
}

@article{becherer2023mean,
  title={Mean-field games of speedy information access with observation costs},
  author={Becherer, Dirk and Reisinger, Christoph and Tam, Jonathan},
  journal={arXiv preprint arXiv:2309.07877},
  year={2023}
}

@article{aqsha2025equilibrium,
  title={Equilibrium reward for liquidity providers in automated market makers},
  author={Aqsha, Alif and Bergault, Philippe and S{\'a}nchez-Betancourt, Leandro},
  journal={arXiv preprint arXiv:2503.22502},
  year={2025}
}

@article{jaimungal2023optimal,
  title={Optimal trading in automatic market makers with deep learning},
  author={Jaimungal, Sebastian and Saporito, Yuri F and Souza, Max O and Thamsten, Yuri},
  journal={arXiv preprint arXiv:2304.02180},
  year={2023}
}

@article{bichuch2024defi,
  title={DeFi Arbitrage in Hedged Liquidity Tokens},
  author={Bichuch, Maxim and Feinstein, Zachary},
  journal={arXiv preprint arXiv:2409.11339},
  year={2024}
}

@article{Oksendal2008,
  title={Optimal stochastic impulse control with delayed reaction},
  author={Oksendal, Bernt and Sulem, Agnes},
  journal={Applied Mathematics and Optimization},
  volume={58},
  number={2},
  pages={243--255},
  year={2008},
  publisher={Springer}
}

@book{protter2005stochastic,
  title={Stochastic Integration and Differential Equations},
  author={Protter, P.E.},
  isbn={9783540003137},
  lccn={89026265},
  series={Stochastic Modelling and Applied Probability},
  year={2005},
  publisher={Springer Berlin Heidelberg}
}

@book{FlemingSoner2006,
  author    = {Wendell H. Fleming and H. Mete Soner},
  title     = {Controlled Markov Processes and Viscosity Solutions},
  publisher = {Springer},
  series    = {Stochastic Modelling and Applied Probability},
  volume    = {25},
  edition   = {2nd},
  year      = {2006},
  address   = {New York},
  doi       = {10.1007/0-387-31071-1}
}

@article{BRUDER20091436,
title = {Impulse control problem on finite horizon with execution delay},
journal = {Stochastic Processes and their Applications},
volume = {119},
number = {5},
pages = {1436-1469},
year = {2009},
issn = {0304-4149},
doi = {https://doi.org/10.1016/j.spa.2008.07.007},
author = {Benjamin Bruder and Huyên Pham}
}

@book{BertsekasShreve1978,
  author    = {Dimitri P. Bertsekas and Steven E. Shreve},
  title     = {Stochastic Optimal Control: The Discrete-Time Case},
  publisher = {Academic Press},
  year      = {1978},
  address   = {New York},
  series    = {Mathematics in Science and Engineering},
  volume    = {139}
}

@book{rudin1976principles,
  title={Principles of Mathematical Analysis},
  author={Rudin, Walter},
  edition={3rd},
  year={1976},
  publisher={McGraw-Hill},
  address={New York}
}

@article{guo_smooth_fit,
author = {Guo, Xin and Wu, Guoliang},
title = {Smooth Fit Principle for Impulse Control of Multidimensional Diffusion Processes},
journal = {SIAM Journal on Control and Optimization},
volume = {48},
number = {2},
pages = {594-617},
year = {2009},
doi = {10.1137/080716001}
}

@article{craishlio92,
  author  = {Crandall, M. and Ishii, H. and Lions, P.-L.},
  title   = {User's guide to viscosity solutions of second order partial differential equations},
  journal = {Bulletin of the American Mathematical Society},
  volume  = {27},
  pages   = {1--67},
  year    = {1992}
}

@article{SEYDEL20093719,
title = {Existence and uniqueness of viscosity solutions for QVI associated with impulse control of jump-diffusions},
journal = {Stochastic Processes and their Applications},
volume = {119},
number = {10},
pages = {3719-3748},
year = {2009},
issn = {0304-4149},
doi = {https://doi.org/10.1016/j.spa.2009.07.004},
author = {Roland C. Seydel},
}

@article{pham1997optimal,
  title={Optimal stopping, free boundary, and American option in a jump-diffusion model},
  author={Pham, Huyên},
  journal={Applied Mathematics and Optimization},
  volume={35},
  number={2},
  pages={145--164},
  year={1997},
  publisher={Springer}
}

@article{oksendal2005,
title = {Optimal stopping with delayed information},
author = {Oksendal, Bernt},
journal = {Stochastics and Dynamics},
volume = {05},
number = {02},
pages = {271-280},
year = {2005},
doi = {10.1142/S0219493705001419}
}

\appendix
\section{Proofs of the Results in Section \ref{sec:viscosity}}
\label{proof_viscosity}
\subsection{Viscosity Solution}
We now provide detailed proofs of Proposition \ref{prop:viscosity_solution}. We first present the proofs on the domain $[0,T)\times\Dc$, and then address separately the case corresponding to the terminal condition at time $T$.

\begin{lemma}
    \label{lem:cont_M_operator}
    The impulse operator $\Mc$ maps $\Cc([0,T]\times \Dc)$ into $\Cc([0,T]\times \Dc)$.
\end{lemma}
\begin{proof}
    Let $\varphi \in \Cc([0,T]\times \Dc)$, $\varepsilon > 0$, $i\in \llbracket 1,K\rrbracket$ and $(t,x,\mathfrak i,\mathfrak v)\in[0,T)\times \Dc$. Then, for any $h\in \R_+^n$ and $\xi\in\Uc$, 
$$-\varepsilon < \varphi(t+\langle \mathfrak h,e_1\rangle,x + h, \mathfrak i + e_i , \mathfrak v + \xi e_i) - \varphi(t,x,\mathfrak i + e_i , \mathfrak v + \xi e_i) < \varepsilon,$$
for a sufficiently small $\|h\| < \delta$ thanks to the continuity of $\varphi$ on $[0,T]\times \Dc$. Hence,
$$\varphi(t,x ,\mathfrak i + e_i , \mathfrak v + \xi e_i) - \varepsilon 
< \varphi(t+\langle \mathfrak h,e_1\rangle , x + h,\mathfrak i + e_i , \mathfrak v + \xi e_i)
< \varphi(t, x ,\mathfrak i + e_i , \mathfrak v + \xi e_i) + \varepsilon.$$
Since $\xi$ is arbitrary, we obtain
$$\Mc\varphi(t,x,\mathfrak i,\mathfrak v) - \varepsilon 
\leq \Mc\varphi(t,x + y,\mathfrak i,\mathfrak v) 
\leq \Mc\varphi(t,x,\mathfrak i,\mathfrak v) + \varepsilon,$$
by taking the infimum for a sufficiently small $\|h\| < \delta$.
\end{proof}
\begin{proposition}
    The value function $v$ is a viscosity super-solution of \eqref{HJB1} and \eqref{HJB2} on $[0,T)\times \Dc$.
\end{proposition}
\begin{proof}
Fix $(t_0,x_0,\mathfrak i_0,\mathfrak v_0)\in[0,T)\times \Dc$. We consider the case where $\sum_{i=1}^N\langle \mathfrak i_0,e_i\rangle<K$. Now let $\varphi\in C^{1,2}([0,T]\times \Dc)$ be such that
$v_*(t_0,x_0, \mathfrak i_0,\mathfrak v_0)=\varphi(t_0,x_0,\mathfrak i_0,\mathfrak v_0)$ and
$v_*-\varphi$ attains a local minimum at this point. In other words, there exists $r_0>0$ such that 
$$v_*(t,x,\mathfrak i_0,\mathfrak v_0) \geq \varphi(t,x,\mathfrak i_0,\mathfrak v_0), \quad \forall (t,x) \in \bar B_{r_0}(t_0,x_0).
$$
By definition of $v_*(t_0, x_0,\mathfrak i_0,\mathfrak v_0)$, there exists a sequence $(t_m, x_m) \in [0, T) \times \R^d$ such that 
$(t_m, x_m) \to (t_0, x_0)$ and $v(t_m, x_m,\mathfrak i_0,\mathfrak v_0) \to v_*(t_0, x_0,\mathfrak i_0,\mathfrak v_0)$ 
as $m \to \infty$. By the continuity of $\varphi$, we also have 
$$\gamma_m := v(t_m, x_m,\mathfrak i_0,\mathfrak v_0) - \varphi(t_m, x_m,\mathfrak i_0,\mathfrak v_0) \underset{m \to +\infty}{\to} 0.$$
  
Set $\tau_0 = t_n$, $\tau_n = +\infty$ for all $n\geq 1$ and choose $(\mathfrak i_0, \xi_0)$ arbitrarily in $\llbracket 1, K \rrbracket \times \mathbb{R}_+$ such that the control 
$\alpha = \big( (\nu_s)_{s \in [0, T]}, (\tau_n, I_n, \xi_n)_{n \ge 1} \big)$
is admissible, i.e. $\alpha \in \mathcal{A}_{K, \mathfrak{i}_0, \mathfrak{v}_0}(t_m)$.\\

Consider the case where $\xi_0 = 0$. Let the trading speed $\nu$ be constant over time and equal to some $a \in \mathbb{R}$. Then $\alpha_m \in \mathcal{A}_{K,\mathfrak{i}_0,\mathfrak{v}_0}(t_m)$. We denote by $X^{t_m, x_m,\alpha_m }_s$ the associated controlled process.  
Let $\tau_m$ be the stopping time defined by $\theta_m := \inf \{ s \ge t_m : X^{t_m, x_m,\alpha_m }_s \notin \bar B_{r_0}(x_m) \}\wedge \frac{\tilde\tau_0}{2}\wedge T$. Let $(h_m)_{m\geq 1}$ be a strictly positive sequence such that $h_m \underset{m \to +\infty}{\to} 0$ and $\gamma_m / h_m \underset{m \to +\infty}{\to} 0.$ We apply the dynamic programming principle (DPP) (see Theorem \ref{thm:DPP}) at time $\tilde{t}_m:=\theta_m \wedge (t_m+h_m) $, we have that 
\begin{equation*}
\resizebox{\textwidth}{!}{$
\begin{aligned}
    v(t_m,x_m,\mathfrak i_0, \mathfrak v_0) \geq 
\mathbb E\bigg[
&\int_{t_m}^{\tilde{t}_m} f\big(s,X_s^{t_m,x_m,\alpha_m },a\big)\,\mathrm \ud s- \sum_{\tilde \tau_n \in (t_m,\tilde{t}_m]} 
c(\tilde{\tau}_n,X^{t_m,x,\alpha_m }_{\tilde{\tau}^-_n},\xi_n,I_n)
+ v(\tilde{t}_m,
X_{\tilde{t}_m}^{t_m,x_m,\alpha_m },
p(\tilde{t}_m,\alpha_m ))
\bigg]-\varepsilon.\end{aligned}
$}\end{equation*}
Since $v\geq v_*\geq \varphi$, we get 
\begin{equation*}
\resizebox{\textwidth}{!}{$
\begin{aligned}\varphi(t_m,x_m,\mathfrak i_0, \mathfrak v_0) + \gamma_m\geq 
\mathbb E\bigg[
&\int_{t_m}^{\tilde{t}_m} f\big(s,X_s^{t_m,x_m,\alpha_m },a\big)\,\mathrm \ud s- \sum_{\tilde \tau_n \in (t_m,\tilde{t}_m]} 
c(\tilde{\tau}_n,X^{t_m,x,\alpha_m }_{\tilde{\tau}^-_n},\xi_n,I_n)
+ \varphi(\tilde{t}_m,
X_{\tilde{t}_m}^{t_m,x_m,\alpha_m },
p(\tilde{t}_m,\alpha_m ))
\bigg]-\varepsilon.\end{aligned}
$}\end{equation*}
 Applying Itô’s formula to the process $s \mapsto \varphi\big(s, X_{s}^{t_m,x_m,\alpha_m }, p(s,\alpha_m )\big)$ on the interval $[t_m,\tilde t_m]$ under $\xi_0 = 0$, and dividing the resulting identity by $h_m$, we obtain
\begin{equation*}
\resizebox{\textwidth}{!}{$
\begin{aligned}
&\frac{\gamma_m}{h_m} 
+ \mathbb{E} \Bigg[ 
    \frac{1}{h_m} \int_{t_m}^{\tilde{t}_m} 
    \bigg( 
        \Big(-\frac{\partial \varphi}{\partial t} 
        - \mathcal{L}^a \varphi\Big)(s, X^{t_m, x_m,\alpha_m }_s,p(s,\alpha_m )) 
        - f\big(s, X^{t_m, x_m,\alpha_m }_s, a\big) +\sum_{i=1}^N\mathds{1}_{\{\langle \mathfrak i(s,\alpha_m ),e_i\rangle>0\}}
\,\ell_i c\big(s, X^{t_m, x_m,\alpha_m }_s,\langle \mathfrak i(s,\alpha_m ),e_i\rangle ,
\langle \mathfrak v(s,\alpha_m ),e_i\rangle\big)\\&- \sum_{i=1}^N\mathds{1}_{\{\langle \mathfrak i(s,\alpha_m ),e_i\rangle>0\}}
\,\ell_i
\Big(
\varphi\big(s,\Gamma(s,X^{t_m, x_m,\alpha_m }_s, \langle \mathfrak v(s,\alpha_m ),e_i\rangle),
\mathfrak i(s,\alpha_m )-\langle \mathfrak i(s,\alpha_m ),e_i\rangle e_i,
\mathfrak v(s,\alpha_m )-\langle \mathfrak v(s,\alpha_m ),e_i\rangle e_i\big)
- \varphi\big(s,X^{t_m, x_m,\alpha_m }_s,p(s,\alpha_m )\big)
\Big)\bigg)\d s 
     \, 
\Bigg]\\&\geq 0,\end{aligned}
$}
 \end{equation*}
after observing that the stochastic integral term vanishes when taking expectations, as its integrand is bounded from the growth conditions in Assumptions \ref{assumptions}. The dynamics of $X^{t_m,x,\alpha_m }$ between $t_m$ and $\tilde{t}_m$ are therefore governed by the SDE
\begin{equation*}
\begin{aligned}
X^{t_m,x_m,\alpha_m }_u &= x + \int_t^u b(s, X^{t_m,x_m,\alpha_m }_s,a)\,\d s + \int_t^u \sigma(s, X^{t_m,x_m,\alpha_m }_s,a)\,\d W_s, \quad \forall u\in [t_m,\tilde{t}_m].
\end{aligned}
 \end{equation*}
 Therefore, the trajectory $X^{t_m, x_m, \alpha_m }_s$ is almost surely continuous $t_m$ and $\tilde{t}_m$. We deduce that for $m$ sufficiently large, it holds that $\theta_m = t_m + h_m$ almost surely. Now let $\Delta N_i^{(m)}:=N^i_{t_m+h_m}-N^i_{t_m}$ be the number of execution opportunities for priority level $i$
on $(t_n,T]$. Hence, we have that $\Delta N_i^{(m)}\sim \mathrm{Poisson}(\ell_i h_m)$, and
$$\mathbb P\Big(\Delta N_i^{(m)}\ge 1\Big)
=1-e^{-\ell_i h_m}\le \ell_i h_m \underset{m \to +\infty}{\to} 0.$$ Since $\xi_0 = 0$, there are no new orders from our control on $(t_m,\tilde t_m]$. Consequently, $$\lim_{m\to +\infty} \mathbb P\Big(p(s,\alpha_m) = (\mathfrak i_0, \mathfrak v_0),~~\forall s\in(t_m,\tilde t_m]\Big) = 1.$$ Applying the mean value theorem, $-\frac{\partial \varphi}{\partial t}(s,X^{t_m, x_m,\alpha_m }_s,\mathfrak i_0, \mathfrak v_0)- H^a(s,X^{t_m, x_m,\alpha_m }_s,\mathfrak i_0, \mathfrak v_0,\varphi,\frac{\partial \varphi}{\partial x},\frac{\partial^2 \varphi}{\partial x^2})$ 
then converges almost surely to 
$$-\frac{\partial \varphi}{\partial t}(t_0,x_0,\mathfrak i_0, \mathfrak v_0)- H^a(t_0,x_0,\mathfrak i_0, \mathfrak v_0,\varphi,\tfrac{\partial \varphi}{\partial x},\tfrac{\partial^2 \varphi}{\partial x^2}),$$ as $m \to +\infty$. In addition, it is uniformly bounded by a constant independent of $m$. Therefore, by the dominated convergence theorem, we obtain
\begin{equation}
    \begin{split}\label{proof:sub1}-\frac{\partial \varphi}{\partial t}(t_0,x_0,\mathfrak i_0, \mathfrak v_0)- H^a(t_0,x_0,\mathfrak i_0, \mathfrak v_0,\varphi,\tfrac{\partial \varphi}{\partial x},\tfrac{\partial^2 \varphi}{\partial x^2})
\ge 0.\end{split}
\end{equation}
The desired result follows from the arbitrariness of $a \in \R$.\\

Let us now examine the case $\xi_0 > 0$. In what follows, we make no specific assumption on the trading speed $\nu$. We apply the DPP again at time $\tilde{t}_m:=\theta_m \wedge (t_m+h_m) $, we have that 
\begin{equation*}
    \begin{split}
    v(t_m,x_m,\mathfrak i_0, \mathfrak v_0) \geq 
\mathbb E\bigg[
&\int_{t_m}^{\tilde{t}_m} f\big(s,X_s^{t_m,x_m,\alpha_m },\nu_s\big)\,\mathrm \ud s - \sum_{\tilde \tau_n \in (t,\tilde{t}_m]} 
c(\tilde{\tau}_n,X^{t,x,\alpha_m }_{\tilde{\tau}^-_n},\xi_n,I_n)
\\&+ v\big(\tilde{t}_m,
X_{\tilde{t}_m}^{t_m,x_m,\alpha_m },
\mathfrak i(\tilde{t}_m,\alpha_m )+ e_{I_0},
\mathfrak v(\tilde{t}_m,\alpha_m ) +\xi_0 e_{I_0}\big)
\bigg]-\varepsilon.\end{split}
\end{equation*}

 Knowing that $v\geq v_*\geq \varphi$ and letting $m$ go to infinity, get that 
\begin{equation*}
    \begin{split}
    v(t_0,x_0,\mathfrak i_0, \mathfrak v_0) \geq 
\mathbb E\bigg[
&\varphi\big(t_0,
x_0,
\mathfrak i_0+ e_{I_0},
\mathfrak v_0 +\xi_0 e_{I_0}\big)
\bigg]-\varepsilon.\end{split}
\end{equation*}
The last inequality hold for every $(\xi,i)\in\Uc\times\setK$. Therefore,
\begin{equation}
\label{proof:sub2}
    \begin{split}
    v(t_0,x_0,\mathfrak i_0, \mathfrak v_0) \geq 
\mathbb E\bigg[
&\sup_{(\xi,I)\in\Uc\times\setK}
\varphi\big(t_0,
x_0,
\mathfrak i_0+ e_{I_0},
\mathfrak v_0 +\xi_0 e_{I_0}\big)\bigg]-\varepsilon\geq \mathcal{M}\varphi(t_0,x_0,\mathfrak i_0, \mathfrak v_0)-\varepsilon.\end{split}
\end{equation}
Combining \eqref{proof:sub1} and \eqref{proof:sub2}, and using the arbitrariness of $\varepsilon$, concludes the proof. The case where $\sum_{i=1}^N \langle \mathfrak{i}_0, e_i \rangle = K$ is handled similarly to the first part of the proof.
\end{proof}

\begin{proposition}
    The value function $v$ is a viscosity sub-solution of \eqref{HJB1} and \eqref{HJB2} on $[0,T)\times \Dc$.
\end{proposition}
\begin{proof} Fix $(t_0,x_0,\mathfrak i_0,\mathfrak v_0)\in[0,T)\times \Dc$. As in the previous case, we detail the proof only when $\sum_{i=1}^N\langle \mathfrak i_0,e_i\rangle<K$, since the arguments for the equality case $\sum_{i=1}^N\langle \mathfrak i_0,e_i\rangle=K$ follow directly by the same reasoning. Let $\varphi\in C^{1,2}([0,T]\times \Dc)$ be such that
$v^*(t_0,x_0, \mathfrak i_0,\mathfrak v_0)=\varphi(t_0,x_0,\mathfrak i_0,\mathfrak v_0)$ and
$v^*-\varphi$ attains a local maximum at this point. In other words, there exists $r_0>0$ such that 
$$v^*(t,x,\mathfrak i_0,\mathfrak v_0) \leq \varphi(t,x,\mathfrak i_0,\mathfrak v_0), \quad \forall (t,x) \in \bar B_{r_0}(t_0,x_0).
$$
If $(v^*-\mathcal{M}v^*)(t_0,x_0,\mathfrak i_0, \mathfrak v_0)\leq 0$, the subsolution condition is immediately satisfied. Now assume that $(v^*-\mathcal{M}v^*)(t_0,x_0,\mathfrak i_0, \mathfrak v_0)> 0$ and that there exists $\eta>0$ such that
$$-\frac{\partial \varphi}{\partial t}(t_0,x_0,\mathfrak i_0, \mathfrak v_0)-\sup_{a\in\R}H^a\Big(t_0,x_0,\mathfrak i_0, \mathfrak v_0,\varphi,\tfrac{\partial \varphi}{\partial x},\tfrac{\partial^2 \varphi}{\partial x^2}\Big)>\eta~~\text{and}~~ \big(\varphi - \mathcal{M}\varphi\big)(t_0,x_0,\mathfrak i_0, \mathfrak v_0)>\eta.$$

Note that, by Assumptions \ref{assumptions}, the functions $b$, $\sigma$, $\Gamma$, $c$ and $f$ are continuous in $(t,x,a)$ and satisfy linear growth conditions. 
The term \begin{equation*}
\resizebox{\textwidth}{!}{$
\begin{aligned}(t,x,a)\mapsto\sum_{i=1}^N
\mathds{1}_{\{\langle \mathfrak i_0,e_i\rangle>0\}}
\,\ell_i
\Big(
\varphi\big(t,\Gamma(t,x, \langle \mathfrak v_0,e_i\rangle),
\mathfrak i_0-\langle \mathfrak i_0,e_i\rangle e_i,
\mathfrak v-\langle \mathfrak v_0,e_i\rangle e_i\big)
- \varphi(t,x,\mathfrak i_0,\mathfrak v_0) + c(t,x,
\langle \mathfrak i_0,e_i\rangle ,
\langle \mathfrak v_0,e_i\rangle)
\Big)\end{aligned}
$}
 \end{equation*} involves only finitely many indices and continuous mappings, and is therefore continuous as well. 
Assumptions \ref{assumptions} guarantee that $\sigma\sigma^\top$ is uniformly elliptic and jointly continuous, implying continuity of the trace term $(t,x,a)\mapsto\frac12\,\mathrm{Tr}\Big(\sigma\sigma^\top(t,x,a)\,\frac{\partial^2 \varphi}{\partial x^2}(t,x,\mathfrak{i}_0,\mathfrak{v}_0)\Big)$. Hence, the differential operator $\mathcal{L}^a\varphi$ is continuous in $(t,x,a)$ for each fixed $(\mathfrak{i}_0,\mathfrak{v}_0)$. Since $\R$ is locally compact, Berge’s maximum theorem (see \cite{rudin1976principles}) ensures that the infinimum over $a$ preserves continuity. Thus, $H$ and $\Mc \varphi$ (see Lemma \ref{lem:cont_M_operator}) are continuous in all their arguments, and $$(t,x)\mapsto-\frac{\partial \varphi}{\partial t}(t,x,\mathfrak i_0, \mathfrak v_0)-\sup_{a\in\R}H^a\Big(t,x,\mathfrak i_0, \mathfrak v_0,\varphi,\tfrac{\partial \varphi}{\partial x},\tfrac{\partial^2 \varphi}{\partial x^2}\Big)~~\text{and}~~(t,x)\mapsto\big(\varphi - \mathcal{M}\varphi\big)(t,x,\mathfrak i_0, \mathfrak v_0)$$ are continuous. Consequently, there exists $0<r_1<r_0$ such that $t_0+r_1<T$ and 
\begin{equation}
\label{eq:strict-ineq}
-\frac{\partial \varphi}{\partial t}(t,x,\mathfrak i_0, \mathfrak v_0)-\sup_{a\in\R}H^a\Big(t,x,\mathfrak i_0, \mathfrak v_0,\varphi,\tfrac{\partial \varphi}{\partial x},\tfrac{\partial^2 \varphi}{\partial x^2}\Big)>\eta~~\text{and}~~ \big(\varphi - \mathcal{M}\varphi\big)(t,x,\mathfrak i_0, \mathfrak v_0)>\eta,\end{equation}
for all $(t,x) \in \bar B_{r_1}(t_0,x_0)$. By definition of the upper semi-continuous enveloppe $v^*(t_0, x_0,\mathfrak i_0,\mathfrak v_0)$, there exists a sequence $(t_m, x_m) \in [0, T) \times \R^d$ such that 
$(t_m, x_m) \to (t_0, x_0)$ and $v(t_m, x_m,\mathfrak i_0,\mathfrak v_0) \to v^*(t_0, x_0,\mathfrak i_0,\mathfrak v_0)$ 
as $m \to +\infty$. By the continuity of $\varphi$, we get
$$\gamma_m := v(t_m, x_m,\mathfrak i_0,\mathfrak v_0) - \varphi(t_m, x_m,\mathfrak i_0,\mathfrak v_0) \underset{m \to +\infty}{\to} 0.$$
Let $\varepsilon>0$ and $(h_m)_{m\geq 1}$ be a strictly positive sequence such that $h_m \underset{m \to +\infty}{\to} 0$ and $\gamma_m / h_m \underset{m \to +\infty}{\to} 0.$ Define $\theta_m$ as the first exit time of the controlled state 
$(s,X_s^{t_m,x_m,\alpha_m^\varepsilon})$
from $\bar B_{r_1}(t_0,x_0)$, truncated at $t_m+h_m$. In other words,
$$\theta_m:=\inf\Big\{s\ge t_m:\ \big(s,X_s^{t_m,x_m,\alpha_m^\varepsilon}\big)\notin\bar B_{r_1}(t_0,x_0)\Big\}\wedge \frac{\tilde\tau_0}{2}\wedge(t_m+h_m)\wedge T.$$ We apply the dynamic programming principle (DPP2) (see Theorem \ref{thm:DPP}) at time $\tilde{t}_m:=\theta_m \wedge (t_m+h_m) $, there exists a control $\alpha^\varepsilon_m := \big((\nu_s)_{s\in[0,T]}, (\tau_n, I_n, \xi_n)_{n \ge 1}\big)$ such that
\begin{equation*}
\resizebox{\textwidth}{!}{$
\begin{aligned}v(t_m,x_m,\mathfrak i_0,\mathfrak v_0)
\le
\E\bigg[\int_{t_m}^{\tilde{t}_m}
f\big(s,X_s^{t_m,x_m,\alpha_m^\varepsilon},\nu_s\big)\,\ud s- \sum_{\tilde \tau_n \in (t_m,\tilde{t}_m]} 
c(\tilde{\tau}_n,X^{t_m,x,\alpha_m^\varepsilon}_{\tilde{\tau}^-_n},\xi_n,I_n)
+v\big(\tilde{t}_m,X_{\tilde{t}_m}^{t_m,x_m,\alpha_m^\varepsilon},p(\tilde{t}_m,\alpha_m^\varepsilon)\big)\bigg]+\varepsilon.\end{aligned}
$}
 \end{equation*}
Subtracting $\varphi(t_m,x_m,\mathfrak i_0,\mathfrak v_0)$ and using $v^*\le\varphi$ on $\bar B_{r_1}(t_0,x_0)$ gives
\begin{equation*}
    \begin{split}
        \gamma_m
\le
&\E\bigg[\int_{t_m}^{\tilde{t}_m}
f\big(s,X_s^{t_m,x_m,\alpha_m^\varepsilon},\nu_s\big)\,\ud s - \sum_{\tilde \tau_n \in (t_m,\tilde{t}_m]} 
c(\tilde{\tau}_n,X^{t_m,x,\alpha_m^\varepsilon}_{\tilde{\tau}^-_n},\xi_n,I_n)
\\&\qquad+\varphi\big(\tilde{t}_m,X_{\tilde{t}_m}^{t_m,x_m,\alpha_m^\varepsilon},p(\tilde{t}_m,\alpha_m^\varepsilon)\big)
-\varphi(t_m,x_m,\mathfrak i_0,\mathfrak v_0)\bigg]+\varepsilon.\end{split}
\end{equation*}
Applying Itô's formula to $\varphi$ on $[t_m,\tilde{t}_m]$ and taking expectations yields  ,
$$\gamma_m
\le
\E\bigg[\int_{t_m}^{\tilde{t}_m}
\frac{\partial \varphi}{\partial t}(s,X_s^{t_m,x_m,\alpha_m^\varepsilon},p(s,\alpha_m^\varepsilon))+H^{\nu_s}\big(s,X_s^{t_m,x_m,\alpha_m^\varepsilon},p(s,\alpha_m^\varepsilon),\varphi,\tfrac{\partial \varphi}{\partial x},\tfrac{\partial^2 \varphi}{\partial x^2}\big)
\,\ud s\bigg]+\varepsilon.$$
Since $(\varphi-\Mc\varphi)>\eta$ on $\bar B_{r_1}(t_0,x_0)$, impulses strictly decrease the continuation value of $\varphi$. 
Thus, for $\varepsilon$ and $m$ small enough, $\alpha_m^\varepsilon$ involves no intervention on $(t_m,\tilde{t}_m]$ and $\P\Big(p(s,\alpha_m^\varepsilon)  \underset{m \to +\infty}{\to} (\mathfrak i_0, \mathfrak v_0),~~ \forall s\in[t_m,\tilde{t}_m]\Big) = 1$. Hence, by \eqref{eq:strict-ineq}, we get
$$\gamma_m
\le
-\eta\,\E[\tilde{t}_m-t_m]+\varepsilon.$$
Since $\E[\tilde{t}_m-t_m]\le h_m$, taking $\varepsilon=\tfrac{\eta}{2}h_m$ gives
$$\gamma_m\le-\tfrac{\eta}{2}h_m.$$
Dividing by $h_m$ and sending $m\to+\infty$ leads to a contradiction. 
\end{proof}

\subsection{Terminal Condition}
We now turn to the analysis of the terminal condition.
\begin{proposition}\label{prop:terminal_cond}
    The value function $v$ defined in \eqref{eq:value_function} satisfies 
    $$
v(T^-,x,\mathfrak i,\mathfrak v) =v(T,x,\mathfrak i,\mathfrak v)=g(x),\quad\forall(x,\mathfrak i,\mathfrak v)\in \Dc.
$$
\end{proposition}
\begin{proof}
 Let $(x,\mathfrak i,\mathfrak v)\in \Dc$ and a sequence $(t_n)_{n\geq 1}$ such that $\underset{n\to +\infty}{\lim} t_n= T$. By the DPP at $\tau=T$, there exists $\alpha_n
 ^\varepsilon\in\mathcal A_{K,\mathfrak i,\mathfrak v}(t_n)$ such that
 \begin{equation}
     \label{ineq_proof_cont_T}
 J(t_n,x,\alpha_n^\varepsilon) - \varepsilon \leq v(t_n,x, \mathfrak i,\mathfrak v)\leq J(t_n,x,\alpha_n^\varepsilon),\end{equation}
 where
$$\begin{aligned}
J(t_n,x,\alpha_n^\varepsilon)
= \E\bigg[
&\int_{t_n}^{T} f(s,X^{t_n,x,\alpha_n^\varepsilon}_s,\nu_s)\ud s
-\sum_{\tilde\tau_m\in(t_n,T]} c(\tilde\tau_m,X^{t_n,x,\alpha_n^\varepsilon}_{\tilde\tau^-_m},\xi_m,I_m)+ v(T,X^{T,x,\alpha_n^\varepsilon}_T,p(T,\alpha_n^\varepsilon))
\bigg].
\end{aligned}$$
Note that, for notational simplicity, we omit writing the dependence of $\big(\nu,(\xi_m,I_m)_{m\geq 1}\big)$ on $\varepsilon$. Set $h_n:=T-t_n\downarrow0$. By the continuity and growth conditions on $f$ (see Assumptions \ref{assumptions}) and the moment bounds for $X^{\alpha}$ (see Proposition \ref{existence_uniqueness_SDE}),
$$\int_{t_n}^{T} f(s,X^{t_n,x,\alpha_n^\varepsilon}_s,\nu_s)\ud s \xrightarrow[n\to\infty]{L^1} 0.$$
Let $\Delta N_i^{(n)}:=N^i_T-N^i_{t_n}$. We have that $\Delta N_i^{(n)}\sim \mathrm{Poisson}(\ell_i h_n)$. Hence,
$$\mathbb P\Big(\Delta N_i^{(n)}\ge 1\Big)
=1-e^{-\ell_i h_n}\le \ell_i h_n \xrightarrow[n\to\infty]{} 0.$$
By admissibility and the
growth bound on $c$, then by linearity of expectation and nonnegativity,
$$\begin{aligned}\E\bigg[\sum_{\tilde\tau_m\in(t_n,T]} c(\tilde\tau_m,X^{t_n,x,\alpha_n^\varepsilon}_{\tilde\tau_m},\xi_m,I_m)\bigg]
&\le \E\bigg[C_c\Big(1+\underset{t_n\leq u\leq T}{\sup}\|X^{t_n,x,\alpha_n^\varepsilon}_u\|+\bar V\Big)\bigg]\E\bigg[\sum_{i=1}^N \Delta N_i^{(n)}\bigg]
\\&= C_c\Bigg(1+\E\bigg[\underset{t_n\leq u\leq T}{\sup}\|X^{t_n,x,\alpha_n^\varepsilon}_u\|\bigg]+\bar V\Bigg)\Big(\sum_{i=1}^N \ell_i\Big)\,h_n \xrightarrow[n\to\infty]{} 0.\end{aligned}$$
Since $\underset{n\to +\infty}{\lim} X^{t_n,x,\alpha}_{T} =  x$ in probability, the dominated convergence yields
$$\lim_{n\to+\infty} J(t_n,x,\alpha_n^\varepsilon) = \E\Big[v(T,X^{T,x,\alpha_n^\varepsilon}_T,p(T,\alpha_n^\varepsilon))\Big] = v(T,x,\mathfrak i,\mathfrak v) = g(x)$$
Hence, by applying inequality \eqref{ineq_proof_cont_T} and noting that $\varepsilon$ can be chosen arbitrarily, we obtain the desired result.
\end{proof}

\subsection{Uniqueness and Continuity Result}
The detailed proofs of Theorem \ref{theo:comparaison} are presented below.
\begin{definition}[Strict supersolution]\label{def: strict_super}
    For $\eta>0$, we say that a family of locally bounded functions $v$ define a viscosity $\eta$- strict supersolution of \eqref{HJB1} and \eqref{HJB2} on $[0,T)\times\Dc$ if it satisfies:
   \begin{enumerate}
       \item For $( t, x, \mathfrak i,\mathfrak v) \in [0,T)\times\Dc$ and any smooth test function $\varphi \in C^{1,2}([0,T]\times\Dc)$ such that $(v_* - \varphi)$ attains a local minimum at $(t, x, \mathfrak i,\mathfrak v)$ over the set $[t, t+\delta) \times B_\delta(x)\times \llbracket 1,K\rrbracket^K\times B_\delta(\mathfrak v) \subset [0,T)\times\Dc$ for some $\delta > 0$, we have 
\begin{align*}
\min \bigg\{&-\frac{\partial \varphi}{\partial t}(t,x,\mathfrak i,\mathfrak v)-\sup_{a\in\R}H^a\Big(t,x,\mathfrak i,\mathfrak v,\varphi,\tfrac{\partial \varphi}{\partial x},\tfrac{\partial^2 \varphi}{\partial x^2}\Big)\,,\, \big(\varphi - \mathcal{M}\varphi\big)(t,x,\mathfrak i,\mathfrak v)\bigg\}> \eta,
\end{align*}
 \item Additionally, for any $(t, x, \mathfrak i,\mathfrak v) \in [0,T)\times\Dc$ and any smooth test function $\varphi \in C^{1,2}([0,T]\times\Dc)$ such that $(v_* - \varphi)$ attains a local minimum at $(t, x, \mathfrak i,\mathfrak v)$ over the set $[t, t+\delta) \times B_\delta(x)\times \llbracket 1,K\rrbracket^K\times B_\delta(\mathfrak v) \subset [0,T)\times\Dc$ for some $\delta > 0$, we have 
\begin{equation*}\begin{split}
    -\frac{\partial \varphi}{\partial t}(t,x,\mathfrak i,\mathfrak v)&- \sup_{a\in\R}H^a\Big(t,x,\mathfrak i,\mathfrak v,\varphi,\tfrac{\partial \varphi}{\partial x},\tfrac{\partial^2 \varphi}{\partial x^2}\Big) > \eta.\end{split}
\end{equation*}
\end{enumerate}
\vspace{-0.3cm}
The first part of the definition covers the case $\sum_{i=1}^N\langle\mathfrak i,e_i\rangle < K$ and the second the case $\sum_{i=1}^N\langle\mathfrak i,e_i\rangle = K$. 
\end{definition}
\begin{lemma}
\label{lem:strict_supersolution}
Let $v:[0,T]\times \Dc\to\mathbb{R}$ be a viscosity supersolution of \eqref{HJB1} and \eqref{HJB2}. Define
$$m(\mathfrak i):=\sum_{j=1}^N \langle \mathfrak i,e_j\rangle \in \{0,\dots,K\}
~~\text{and}~~
\bar\Lambda := \sum_{j=1}^N \ell_j .$$
Then, for any $\eta>0$, there exists an $\eta$-strict viscosity supersolution $v^\eta$ of \eqref{HJB1} and \eqref{HJB2} such that
$$v^\eta(t,x,\mathfrak i,\mathfrak v)
= v(t,x,\mathfrak i,\mathfrak v) + \eta\,\phi_1(t,\mathfrak i) + \eta\,\phi_2(t,x),$$
with $(t,x,\mathfrak i,\mathfrak v)\in[0,T]\times \Dc$ and
$$\phi_1(t,\mathfrak i) := (1+K \bar\Lambda)(T-t) + (K-m(\mathfrak i))
~~\text{and}~~
\phi_2(t,x) := \frac12 e^{L(T-t)}\bigl(1+\|x\|^2\bigr),$$
where $L>0$. Additionally, there exist constants $C_1,C_2>0$, independent of $\eta$, such that
\begin{equation}
\label{eq:strict_growth}
v(t,x,\mathfrak i,\mathfrak v)+\eta\,C_1\|x\|^2
\;\le\;
v^\eta(t,x,\mathfrak i, \mathfrak v)
\;\le\;
v(t,x,\mathfrak i,\mathfrak v)+\eta\,C_2(1+\|x\|^2).
\end{equation}
\end{lemma}
\begin{proof}
Fix $(t,x,\mathfrak i,\mathfrak v)\in[0,T]\times \Dc$ and $\eta>0$. As in the previous case, we detail the proof only when $\sum_{i=1}^N\langle \mathfrak i,e_i\rangle<K$. Let $\varphi^\eta\in C^{1,2}([0,T]\times \Dc)$ be such that
$v^{\eta}_*(t,x, \mathfrak i,\mathfrak v)=\varphi^\eta(t,x,\mathfrak i,\mathfrak v)$ and
$v^{\eta}_*-\varphi^\eta$ attains a local maximum at this point. In other words, there exists $r_0>0$ such that 
$$v^{\eta}_*(t',x',\mathfrak i,\mathfrak v) \leq \varphi^\eta(t',x',\mathfrak i,\mathfrak v), \quad \forall (t',x') \in \bar B_{r_0}(t,x).
$$ Let $\varphi: [0,T]\times\Dc\to \R$ be defined as follows $$\varphi(t',x',\mathfrak i,\mathfrak v)
:= \varphi^\eta(t',x',\mathfrak i,\mathfrak v) - \eta\,\phi_1(t',\mathfrak i) - \eta\,\phi_2(t',x'), \quad \forall (t',x') \in \bar B_{r_0}(t,x).$$
Note that $\varphi^\eta\in C^{1,2}([0,T]\times \Dc)$,
$v^*(t,x, \mathfrak i,\mathfrak v)=\varphi(t,x,\mathfrak i,\mathfrak v)$ and
$v^*-\varphi$ attains a local maximum at $(t,x, \mathfrak i,\mathfrak v)$. For any admissible impulse $(\xi, j)\in\Uc\times\setK$ and $(t',x') \in \bar B_{r_0}(t,x)$,
$$\phi_1(t',\mathfrak i)-\phi_1(t',\mathfrak i+e_j)
= (K-m(\mathfrak i))-(K-m(\mathfrak i)-1)=1,$$
while $\phi_2$ is independent of $(\mathfrak i,\mathfrak v)$. Hence,
$$\varphi^\eta(t',x',\mathfrak i,\mathfrak v)
-
\varphi^\eta(t',x',\mathfrak i+e_j,\mathfrak v+\xi e_j)
=
\varphi(t',x',\mathfrak i,\mathfrak v)
-
\varphi(t',x',\mathfrak i+e_j,\mathfrak v+\xi e_j)
+\eta, \quad \forall (t',x') \in \bar B_{r_0}(t,x).$$
Since $v$ is a viscosity supersolution, $\varphi\ge \Mc \varphi$, and taking the supremum over
$(\xi, j)$ gives \begin{equation}\label{first_ineq_strict_super}
\varphi^\eta(t',x',\mathfrak i,\mathfrak v) \ge \Mc \varphi^\eta(t',x',\mathfrak i,\mathfrak v) + \eta, \quad \forall (t',x') \in \bar B_{r_0}(t,x).
\end{equation}   
Additionally, we have that 
$$-\frac{\partial \varphi}{\partial t}(t',x',\mathfrak i, \mathfrak v)-\sup_{a\in\R}H^a\Big(t',x',\mathfrak i, \mathfrak v,\varphi,\tfrac{\partial \varphi}{\partial x},\tfrac{\partial^2 \varphi}{\partial x^2}\Big)\geq 0, \quad \forall (t',x') \in \bar B_{r_0}(t,x).$$
Note that $\mathcal L^a\phi_1=0$ for all $a\in \R$ and $\frac{\partial \phi_1}{\partial t}=-(1+K\bar\Lambda)$, while the
execution operator satisfies
$$\Jc\phi_1(t',x',\mathfrak i,\mathfrak v)
=
\sum_{j=1}^N \1_{\{\langle \mathfrak i,e_j\rangle>0\}}\ell_j\Big(\langle \mathfrak i,e_j\rangle-c(t',x',
\langle \mathfrak i,e_i\rangle ,
\langle \mathfrak v,e_i\rangle)\Big)
\le K\bar\Lambda, \quad \forall (t',x') \in \bar B_{r_0}(t,x).$$
Thus $-\frac{\partial \phi_1}{\partial t} - \sup_{a\in\R}\mathcal L^a\phi_1 - \Jc\phi_1 \ge 1.$ Second, using Assumptions \ref{assumptions}, there exists $C>0$ such that
\begin{equation} \label{growth_ineq_strict}
\sup_{a\in\R} \mathcal L^a \phi_2 + \Jc\phi_2
\le C e^{L(T-t')}(1+\|x'\|+\|x'\|^2)~~
\text{and}~~
-\frac{\partial \phi_2}{\partial t}
=
\frac{L}{2}e^{L(T-t')}(1+\|x'\|^2),\end{equation}
for all $(t',x') \in \bar B_{r_0}(t,x)$. Choosing $L$ large enough yields
$-\frac{\partial \phi_2}{\partial t}
- \sup_{a\in\R}\mathcal L^a \phi_2
-
J\phi_2
\ge 0.$ Combining the two estimates shows that
$$-\frac{\partial}{\partial t}(\phi_1+\phi_2)
-
\sup_{a\in\R} \mathcal \Lc^a(\phi_1+\phi_2)
-
\Jc(\phi_1+\phi_2)
\ge 1,$$
which implies, after multiplication by $\eta$, that 
$$-\frac{\partial \varphi^\eta}{\partial t}(t',x',\mathfrak i, \mathfrak v)-\sup_{a\in\R}H^a\Big(t',x',\mathfrak i, \mathfrak v,\varphi^\eta,\tfrac{\partial \varphi^\eta}{\partial x},\tfrac{\partial^2 \varphi^\eta}{\partial x^2}\Big)\geq \eta, \quad \forall (t',x') \in \bar B_{r_0}(t,x).$$
Hence, $v^\eta$ is a strict supersolution of \eqref{HJB1} and \eqref{HJB2}. The growth result follows directly from \eqref{growth_ineq_strict}.
\end{proof}
\begin{theorem} If $w$ is a viscosity subsolution of \eqref{HJB1} and \eqref{HJB2} and
$v$
is a viscosity supersolution of \eqref{HJB1} and \eqref{HJB2}, such that 
$$w^*(T, x,\mathfrak{i},\mathfrak{v})\leq
v_*(T, x,\mathfrak{i},\mathfrak{v}),$$
for all $(t,x,\mathfrak{i},\mathfrak{v})\in [0,T]\times\Dc$, then $w \leq v$ on $[0,T]\times\Dc$. 
\end{theorem}
\begin{proof}
Let $w$ be a viscosity subsolution and $v$ a viscosity supersolution of
\eqref{HJB1} and \eqref{HJB2} in the sense of Definition \ref{viscosity_def}, and assume that
\begin{equation}
    \label{cpm:ineq1}
w^*(T,x,\mathfrak i,\mathfrak v) \le v_*(T,x,\mathfrak i,\mathfrak v), \quad \forall (x,\mathfrak i,\mathfrak v)\in \Dc.\end{equation}
Our goal is to show that $\varrho := \sup_{(t,x,\mathfrak i,\mathfrak v)\in [0,T]\times \Dc} v^{\eta}_*(t,x,\mathfrak i,\mathfrak v) - w^*(t,x,\mathfrak i,\mathfrak v)\leq 0$, where $v^{\eta}_*$ has been introduced in Lemma \ref{lem:strict_supersolution} for $\eta>0$. Assume by contradiction that $\varrho > 0$. Using the growth results in Lemma \ref{lem:growth_value_function_LQ}, we get that 
$$v^{\eta}_*(t,x,\mathfrak i,\mathfrak v) - w^*(t,x,\mathfrak i,\mathfrak v)\leq C_1 (1 + \|x\|) - C_2 \|x\|^2,\quad \forall (t,x,\mathfrak i,\mathfrak v)\in [0,T]\times\Dc.$$
In particular, $\lim_{\|x\|\to +\infty}v^{\eta}_*(t,x,\mathfrak i,\mathfrak v) - w^*(t,x,\mathfrak i,\mathfrak v) = -\infty$. Additionally, by \eqref{cpm:ineq1}, the supremum is cannot be attained at terminal time $T$. Therefore, the supremum of $v^{\eta}_*-w^*$ is attained an interior point $(t_0,x_0,\mathfrak i_0, \mathfrak v_0) \in \Oc\subset[0,T)\times B_{x_0}(\bar X)\times\I\times \V\subset[0,T]\times \Dc$, with $\bar X>0$. In other words, we have $\varrho
= v^{\eta}_*(t_0,x_0,\mathfrak i_0, \mathfrak v_0)-w^*(t_0,x_0,\mathfrak i_0, \mathfrak v_0)$. Let $k\geq 1$ and define, for all $(t,x,x',\mathfrak i,\mathfrak v)\in[0,T]\times \R^d\times\Dc$, $$F_k(t,x,x',\mathfrak i,\mathfrak v) := v^{\eta}_*(t,x,\mathfrak i,\mathfrak v) - w^*(t,x',\mathfrak i,\mathfrak v)
-d_k(x,x'),$$ where $d_k(x,x'):= \frac{k}{2}\big(\|x-x'\|^2+\|x-x'\|^4\big)$. Moreover, define $$\varrho_k :=
\underset{(t,x,\mathfrak i,\mathfrak v)\in [0,T]\times\R^d\times \Dc}{\sup} F_k(t,x,x',\mathfrak i,\mathfrak v).$$ 
Since $F_k$ is upper semi-continuous and coercive, its supremum is
attained, for all $k\in \N$, at some point $(\hat{t}_k,\hat{x}_k, \hat{x}'_k,\hat{\mathfrak i}_k,\hat{\mathfrak v}_k)
\in [0,T]\times\R^d\times \Dc$. By means of the Bolzano–Weierstrass theorem, there exists a subsequence $\big(\hat{t}_{n_k},\hat{x}_{n_k}, \hat{x}'_{n_k},\hat{\mathfrak i}_{n_k},\hat{\mathfrak v}_{n_k}\big)_{k\geq 0}$ that converges to a point
$(\hat{t}_0,\hat{x}_0, \hat{x}'_0,\hat{\mathfrak i}_0,\hat{\mathfrak v}_0)$ as $k \to +\infty$. In the following, we will continue using $k$ as an index instead of $n_k$ to avoid the proliferation of indices. For $k$ large
enough, we can then assume that $\hat{t}_k < T$. By definition of $F_k$, we have the following inequality
$$ F_k(\hat{t}_0,\hat{x}_0, \hat{x}_0,\hat{\mathfrak i}_0,\hat{\mathfrak v}_0) \leq
F_k(\hat{t}_k,\hat{x}_k, \hat{x}'_k,\hat{\mathfrak i}_k,\hat{\mathfrak v}_k).$$
In particular, we have
\begin{equation*}
\frac{k}2\big(\|\hat{x}_k-\hat{x}'_k\|^2+\|\hat{x}_k-\hat{x}'_k\|^4\big) \leq - v^{\eta}_*(\hat{t}_0,\hat{x}_0,\hat{\mathfrak i}_0,\hat{\mathfrak v}_0)
+ w^*(\hat{t}_0,\hat{x}_0,\hat{\mathfrak i}_0,\hat{\mathfrak v}_0) + v^{\eta}_*(\hat{t}_k,\hat{x}_k, \hat{\mathfrak i}_k,\hat{\mathfrak v}_k) -
w^*(\hat{t}_k, \hat{x}'_k,\hat{\mathfrak i}_k,\hat{\mathfrak v}_k).
\end{equation*} As $v^{\eta}_*$ and $w^*$ are continuous on the compact set $\Oc$, there exists $C>0$ such that
\begin{equation}
\label{comparaison_diff_limit}
\|\hat{x}_k-\hat{x}'_k\|^2+\|\hat{x}_k-\hat{x}'_k\|^4  \leq  \frac{C}k. 
\end{equation}
Letting $k$ go to $+\infty$, we find $\hat{x}_0 = \hat{x}'_0$. Finally, we show that $\varrho_k$ tends to $\varrho$ when $k$ goes to $+\infty$. Note that $$\varrho=v^{\eta}_*(t_0,x_0,\mathfrak i_0, \mathfrak v_0)-w^*(t_0,x_0,\mathfrak i_0, \mathfrak v_0)= 
H_k(t_0, x_0, x_0, \mathfrak i_0, \mathfrak v_0)\leq H_k(\hat{t}_k,\hat{x}_k,\hat{x}'_k,\mathfrak i_k, \mathfrak v_k).$$ Therefore, $\varrho
\leq \varrho_k$. Moreover, we
have
\begin{equation*}
\varrho_k = v^{\eta}_*(\hat{t}_k,\hat{x}_k, \hat{\mathfrak i}_k,\hat{\mathfrak v}_k) - w^*(\hat{t}_k, \hat{x}'_k,\hat{\mathfrak i}_k,\hat{\mathfrak v}_k) -
\frac{k}{2} \|\hat{x}_k - \hat{x}'_k\|^4 \leq v^{\eta}_*(\hat{t}_k,\hat{x}_k,\hat{\mathfrak i}_k,\hat{\mathfrak v}_k) -
w^*(\hat{t}_k, \hat{x}'_k,\hat{\mathfrak i}_k,\hat{\mathfrak v}_k).
\end{equation*} Since $v^{\eta}_*$ and $w^*$ are upper and lower semi-continuous on $[0,T]\times\Dc$, we get that $$\lim_{k\rightarrow+\infty}v^{\eta}_*(\hat{t}_k, \hat{x}_k,\hat{\mathfrak i}_k,\hat{\mathfrak v}_k) -
w^*(\hat{t}_k, \hat{x}'_k,\hat{\mathfrak i}_k,\hat{\mathfrak v}_k) = v^{\eta}_*(\hat{t}_0, \hat{x}'_0,\hat{\mathfrak i}_0,\hat{\mathfrak v}_0) -
w^*(\hat{t}_0, \hat{x}'_0,\hat{\mathfrak i}_0,\hat{\mathfrak v}_0) \leq \varrho.$$ We conclude that $\lim_{k\rightarrow+\infty}\varrho_k = \varrho$ and
$\lim_{k\rightarrow+\infty} \|\hat{x}_k - \hat{x}'_k\|^4 = 0$. Moreover, we
have $$v^{\eta}_*(\hat{t}_0, \hat{x}'_0,\hat{\mathfrak i}_0,\hat{\mathfrak v}_0) - w^*(\hat{t}_0, \hat{x}'_0,\hat{\mathfrak i}_0,\hat{\mathfrak v}_0) = \varrho.$$
Applying Theorem $3.2$ from \cite{craishlio92} at the point 
$(\hat{t}_k,\hat{x}_k, \hat{x}'_k,\hat{\mathfrak i}_k,\hat{\mathfrak v}_k)$ yields the existence of two symmetric matrices $M_k, M'_k \in \mathbb{R}^{d\times d}$ in the superjet set $J^{2,+}$ of $v^{\eta}_*$ and the subjet set $J^{2,-}$ of $w^*$ such that
$$\Big(\frac{\partial d_k}{\partial
x},M_k\Big)\in J^{2,+}v^{\eta}_*(\hat{t}_k,\hat{x}_k,\hat{\mathfrak i}_k,\hat{\mathfrak v}_k)~~\textrm{and}~~\Big(-\frac{\partial d_k}{\partial
x'},M'_k\Big)\in J^{2,-}w^*(\hat{t}_k, \hat{x}'_k,\hat{\mathfrak i}_k,\hat{\mathfrak v}_k)$$
and the following inequality holds
 \begin{equation}\label{second_order_ineq} \begin{split} 
 \left(\begin{array}{cc}
        M_k & 0 \\
        0 & -M'_k \\
\end{array}
\right) \leq A + \frac{1}{k} A^2, \textrm{  with }A  = D^2 d_k(\hat{x}_k, \hat{x}'_k) = \begin{pmatrix}
k I & -k I \\
- k I & k I
\end{pmatrix}\end{split}, \end{equation} where $I$ is the identity matrix. Using the relationship between superjets along with the supersolution properties of $v^{\eta}_*$ established in Lemma \ref{lem:strict_supersolution}, we deduce from Ishii's Lemma that
\begin{equation}
\label{inegalitesishii1}
\eta \leq \min\Big\{-\sup_{a\in\R}H^a\Big(\hat{t}_k,\hat{x}_k,\hat{\mathfrak i}_k,\hat{\mathfrak v}_k,v^{\eta}_*,\tfrac{\partial d_k}{\partial x},M_k\Big),\;
\big(v^{\eta}_*-\mathcal M v^{\eta}_*\big)(\hat{t}_k,\hat{x}_k,\hat{\mathfrak i}_k,\hat{\mathfrak v}_k)\Big\},
\end{equation}
\begin{equation}
\label{inegalitesishii2}
0 \geq \min\Big\{-\sup_{a\in\R}H^a\Big(\hat{t}_k,\hat{x}'_k,\hat{\mathfrak i}_k,\hat{\mathfrak v}_k,w^*,-\tfrac{\partial d_k}{\partial x'},M'_k\Big),\;
\big(w^*-\mathcal M w^*\big)(\hat{t}_k,\hat{x}'_k,\hat{\mathfrak i}_k,\hat{\mathfrak v}_k)\Big\},
\end{equation}
\begin{remark}
    We use the local definition of viscosity sub- and supersolutions, where test
functions touch the candidate solution locally.
Under the present assumptions, this notion is equivalent to the global
definition for HJB-QVI inequalities. We refer to \cite{SEYDEL20093719} for a detailed discussion.
\end{remark}
It follows from the supersolution property that 
\begin{equation}\label{eq:continuation_region}
v^{\eta}_*(t_0,x_0, \mathfrak i_0, \mathfrak v_0) \ge \Mc v^{\eta}_*(t_0,x_0, \mathfrak i_0, \mathfrak v_0).
\end{equation}
\paragraph*{Case 1:} $w^*(\hat{t}_k,\hat{x}'_k, \hat{\mathfrak i}_k, \hat{\mathfrak v}_k) \leq \Mc w^*(\hat{t}_k,\hat{x}'_k, \hat{\mathfrak i}_k, \hat{\mathfrak v}_k)$. Then
$$\Delta := w^*(\hat{t}_k,\hat{x}'_k, \hat{\mathfrak i}_k, \hat{\mathfrak v}_k) - v^{\eta}_*(\hat{t}_k,\hat{x}_k, \hat{\mathfrak i}_k, \hat{\mathfrak v}_k)\leq \Mc w^*(\hat{t}_k,\hat{x}'_k, \hat{\mathfrak i}_k, \hat{\mathfrak v}_k) -\Mc v^{\eta}_*(\hat{t}_k,\hat{x}_k, \hat{\mathfrak i}_k, \hat{\mathfrak v}_k).$$
By the definition of $\Mc$, for $\varepsilon>0$, there exists $(\xi^\ast,j^\ast)\in \Uc\times \llbracket 1, K\rrbracket$
such that
$$\Mc w^*(\hat{t}_k,\hat{x}'_k, \hat{\mathfrak i}_k, \hat{\mathfrak v}_k)
\leq w^*(\hat{t}_k,\hat{x}'_k, \hat{\mathfrak i}_k+e_{j^\ast}, \hat{\mathfrak v}_k+\xi^\ast e_{j^\ast}) - \varepsilon.$$
Additionally, we have that
$$\Mc v^{\eta}_*(\hat{t}_k,\hat{x}_k, \hat{\mathfrak i}_k, \hat{\mathfrak v}_k)
\geq v^{\eta}_*(\hat{t}_k,\hat{x}_k, \hat{\mathfrak i}_k+e_{j^\ast}, \hat{\mathfrak v}_k+\xi^\ast e_{j^\ast}) $$
Hence
$$\Delta
\le w^*(\hat{t}_k,\hat{x}'_k, \hat{\mathfrak i}_k+e_{j^\ast}, \mathfrak{v}_k+\xi^\ast e_{j^\ast})
   - v^{\eta}_*(\hat{t}_k,\hat{x}_k, \hat{\mathfrak i}_k+e_{j^\ast}, \hat{\mathfrak v}_k+\xi^\ast e_{j^\ast}) - \varepsilon
\le \Delta - \varepsilon,$$
which leads to a contradiction.

\paragraph*{Case 2:} $w^*(\hat{t}_k,\hat{x}'_k, \hat{\mathfrak i}_k, \hat{\mathfrak v}_k) > \Mc w^*(\hat t_k,\hat x'_k, \hat{\mathfrak i}_k, \hat{\mathfrak v}_k)$. It follows from inequalities \eqref{inegalitesishii1} and \eqref{inegalitesishii2} that
\begin{equation}\label{ineq_comp3}
\eta \leq -\sup_{a\in\R}H^a\Big(\hat{t}_k,\hat{x}_k,\hat{\mathfrak i}_k,\hat{\mathfrak v}_k,v^{\eta}_*,\tfrac{\partial d_k}{\partial x},M_k\Big),~~\text{and}~~
0 \geq -\sup_{a\in\R}H^a\Big(\hat{t}_k,\hat{x}'_k,\hat{\mathfrak i}_k,\hat{\mathfrak v}_k,w^*,-\tfrac{\partial d_k}{\partial x'},M'_k\Big).
\end{equation}
Fix $R>0$ and introduce the truncated Hamiltonian
$$H_R := \sup_{|a|\le R} H^a.$$
By replacing $\sup_{a\in \R} H^a$ with $H_R$ in the above inequalities, the supremum is taken
over a compact set and the comparison argument applies similarly. By definition of the supremum and using \eqref{ineq_comp3}, there exists $a^\epsilon_k:=a^\epsilon\Big(\hat{t}_k,\hat{x}'_k,\hat{\mathfrak i}_k,\hat{\mathfrak v}_k,w^*,-\tfrac{\partial d_k}{\partial x'},M'_k\Big)\in [-R,R]$ for each $k\in\N$ such that
\begin{equation*}
    \begin{split}
        \eta&\leq -H^{a^\epsilon_k}\Big(\hat{t}_k,\hat{x}_k,\hat{\mathfrak i}_k,\hat{\mathfrak v}_k,v^\eta_*,\tfrac{\partial d_k}{\partial x},M_k\Big)+ H^{a^\epsilon_k}\Big(\hat{t}_k,\hat{x}'_k,\hat{\mathfrak i}_k,\hat{\mathfrak v}_k,w^*,-\tfrac{\partial d_k}{\partial x'},M'_k\Big)-\epsilon.
    \end{split}
\end{equation*}
Moreover, since the above inequality remains valid when the supremum in
the Hamiltonian is restricted to a sufficiently large compact set,
the sequence $(a_k^\epsilon)_{k\in\N}$ may be chosen bounded. By the Bolzano-Weierstrass theorem, there exists a subsequence,
still denoted $(a_k^\epsilon)_{k\in\N}$, converging to some $a^\epsilon_0\in\R$.
Sending $k\to+\infty$ is therefore justified for fixed $R$. Letting $R\to+\infty$ and using the monotone convergence
$H_R\underset{R\to +\infty}{\to} H$ (see \cite[Chapter~III]{FlemingSoner2006}),
we recover the inequality for the original Hamiltonian. Sending $k$ to $+\infty$, we get by continuity of $b$ and $\sigma$ and inequality \eqref{second_order_ineq} that
\begin{equation*}
    \begin{split}
        \eta&\leq \lim_{k\to+\infty}\frac{1}{2}\Big\langle\left(\begin{array}{cc}M_k& 0\\ 0 &
-M^\prime_k\end{array}\right)\big(\sigma(\hat{t}_k,\hat{x}_k,a^\epsilon_k),\sigma(\hat{t}_k,\hat{x}'_k,a^\epsilon_k)\big)^\top,\big(\sigma(\hat{t}_k,\hat{x}_k,a^\epsilon_k),\sigma(\hat{t}_k,\hat{x}'_k,a^\epsilon_k)\big)\Big\rangle
        \\&\qquad+\mathcal J (w^*-v^\eta_*)(\hat{t}_0,\hat{x}_0,\hat{\mathfrak i}_0,\hat{\mathfrak v}_0)-\epsilon
        \\&\leq \lim_{k\to+\infty} \frac{1}{2}\Big\langle\big(A+\frac{1}{k}
A^2\big)\big(\sigma(\hat{t}_k,\hat{x}_k,a^\epsilon_k),\sigma(\hat{t}_k,\hat{x}'_k,a^\epsilon_k)\big)^\top,\big(\sigma(\hat{t}_k,\hat{x}_k,a^\epsilon_k),\sigma(\hat{t}_k,\hat{x}'_k,a^\epsilon_k)\big)\Big\rangle \\&= \lim_{k\to+\infty}\frac{3}{2}k\big(\sigma(\hat{t}_k,\hat{x}_k,a^\epsilon_k)-\sigma(\hat{t}_k,\hat{x}'_k,a^\epsilon_k)\big)^\top\big(\sigma(\hat{t}_k,\hat{x}_k,a^\epsilon_k)-\sigma(\hat{t}_k,\hat{x}'_k,a^\epsilon_k)\big)\\&\leq \frac{3}{2}C.
    \end{split}
\end{equation*}
This leads to a contradiction since $\eta$ is arbitrary.
\end{proof}
\begin{lemma}
The viscosity solution of \eqref{HJB1} and \eqref{HJB2} is unique and continuous on $[0,T]\times\Dc$.
\end{lemma}
\begin{proof}
   Let $u$ and $\tilde u$ be viscosity solutions of \eqref{HJB1} and \eqref{HJB2} with the same terminal condition. The comparison principle applied to $(u,\tilde u)$ and $(\tilde u,u)$ yields $u=\tilde u$, so the solution is unique. Let $u$ be the unique solution. The envelopes $u^{\ast}$ and $u_{\ast}$ are respectively a subsolution and a supersolution with identical terminal data. Comparison gives $u^{\ast}\le u_{\ast}$, while $u_{\ast}\le u^{\ast}$ holds. Hence $u^{\ast}=u_{\ast}=u$ and $u$ is continuous on $[0,T]\times\Dc$.
\end{proof}

\end{document}